\pdfoutput=1

\documentclass{article}

\usepackage[utf8]{inputenc}

\def\showauthornotes{0}
\def\showkeys{0}
\def\showdraftbox{0}
\def\confversion{0}

\usepackage{xspace,enumerate}
\usepackage{amsmath,amssymb}
\usepackage{amsthm}
\ifnum\confversion=0
	\usepackage[toc,page]{appendix}
\fi
\usepackage{thmtools}
\usepackage{thm-restate}
\usepackage{color,graphicx}
\usepackage{boxedminipage}
\usepackage{makecell}
\usepackage{tabularx}
\usepackage{relsize}

\ifnum\showkeys=1
\usepackage[color]{showkeys}
\fi

\definecolor{darkred}{rgb}{0.5,0,0}
\definecolor{darkgreen}{rgb}{0,0.35,0}
\definecolor{darkblue}{rgb}{0,0,0.55}

\usepackage[dvipsnames]{xcolor}

\ifnum\confversion=1
\usepackage[pdfstartview=FitH,pdfpagemode=UseNone,colorlinks=false,linkcolor=darkblue,filecolor=darkred,citecolor=darkgreen,urlcolor=darkred,pagebackref,bookmarks=false]{hyperref}
\else
\usepackage[pdfstartview=FitH,pdfpagemode=UseNone,colorlinks,linkcolor=NavyBlue,filecolor=blue,citecolor=OliveGreen,urlcolor=NavyBlue,pagebackref]{hyperref}
\fi

\usepackage[capitalise,nameinlink]{cleveref}
\usepackage[T1]{fontenc}
\usepackage{mathtools,dsfont,bbm}
\usepackage{mathpazo}
\usepackage{microtype}
\usepackage[top=1in, bottom=1in, left=1.25in, right=1.25in]{geometry}

\ifnum\showauthornotes=1
\newcommand{\Authornote}[3]{{\sf\small\color{#3}{[#1: #2]}}}
\newcommand{\Authorcomment}[2]{{\sf \small\color{gray}{[#1: #2]}}}
\newcommand{\Authorfnote}[2]{\footnote{\color{red}{#1: #2}}}
\else
\newcommand{\Authornote}[3]{}
\newcommand{\Authorcomment}[2]{}
\newcommand{\Authorfnote}[2]{}
\fi

\ifnum\showdraftbox=1
\newcommand{\draftbox}{\begin{center}
  \fbox{%
    \begin{minipage}{2in}%
      \begin{center}%
        \begin{Large}%
          \textsc{Working Draft}%
        \end{Large}\\
        Please do not distribute%
      \end{center}%
    \end{minipage}%
  }%
\end{center}
\vspace{0.2cm}}
\else
\newcommand{\draftbox}{}
\fi

\newtheorem{theorem}{Theorem}[section]

\newtheorem{observation}[theorem]{Observation}
\newtheorem{definition}[theorem]{Definition}

\newtheorem{lemma}[theorem]{Lemma}
\newtheorem{remark}[theorem]{Remark}

\newtheorem{claim}[theorem]{Claim}

\newtheorem{algo}[theorem]{Algorithm}

\def\FullBox{\hbox{\vrule width 6pt height 6pt depth 0pt}}

\def\qed{\ifmmode\qquad\FullBox\else{\unskip\nobreak\hfil
\penalty50\hskip1em\null\nobreak\hfil\FullBox
\parfillskip=0pt\finalhyphendemerits=0\endgraf}\fi}

\def\qedsketch{\ifmmode\Box\else{\unskip\nobreak\hfil
\penalty50\hskip1em\null\nobreak\hfil$\Box$
\parfillskip=0pt\finalhyphendemerits=0\endgraf}\fi}

\def\eps{\varepsilon}
\def\epsilon{\varepsilon}

\def\eps{\epsilon}

\def\phi{\varphi}
\def\cal{\mathcal}

\def\implies{\Rightarrow}
\newcommand{\sub}{\ensuremath{\subseteq}}

\newcommand{\defeq}{:=}

\newcommand{\given}{\;\ifnum\currentgrouptype=16 \middle\fi \vert\;}

\newcommand{\ie}{i.e.,\xspace}

\newcommand{\etal}{et al.\xspace}

\newcommand{\mper}{\,.}
\newcommand{\mcom}{\,,}

\newcommand{\R}{{\mathbb R}}
\newcommand{\E}{{\mathbb E}}

\newcommand{\N}{{\mathbb{N}}}

\usepackage{nicefrac}

\let\nfrac=\nicefrac

\newcommand{\abs}[1]{\ensuremath{\left\lvert #1 \right\rvert}}

\newcommand{\norm}[1]{\ensuremath{\left\lVert #1 \right\rVert}}

\newcommand{\ip}[2] {\ensuremath{\left\langle #1 , #2 \right\rangle}}

\newcommand{\one}{{\mathbf{1}}}

\newcommand{\Esymb}{\mathbb{E}}
\newcommand{\Psymb}{\mathbb{P}}

\newcommand{\Varsymb}{\mathrm{Var}}
\DeclareMathOperator*{\ExpOp}{\Esymb}

\makeatletter
\def\Pr#1{%
    \ProbabilityRender{\Psymb}{#1}%
}

\def\Ex#1{%
    \ProbabilityRender{\Esymb}{#1}%
}

\def\tildeEx#1{%
    \ProbabilityRender{\widetilde{\Esymb}}{#1}%
}

\def\condPE#1#2{%
	\@ifnextchar\bgroup
	{\ConditionalProbabilityRender{\widetilde{\Esymb}}{#1}{#2}}
	{\ProbabilityRender{\widetilde{\Esymb}}{#1 \given #2}}
}

\def\tildeVar#1{
	\ProbabilityRender{\widetilde{\Varsymb}}{#1}
}

\def\tildeCov#1#2{
	\ProbabilityRender{\tildecov}{#1,#2}
}

\def\ConditionalProbabilityRender#1#2#3#4{
	\renderwithdist{#1}{#2}{#3 \given #4}	
}

\def\ProbabilityRender#1#2{
  \@ifnextchar\bgroup%
  {\renderwithdist{#1}{#2}}
   {\singlervrender{#1}{#2}}
}
\def\singlervrender#1#2{%
   \ensuremath{\mathchoice
       {{#1}\left[ #2 \right]}
       {{#1}[ #2 ]}
       {{#1}[ #2 ]}
       {{#1}[ #2 ]}
   }
}
\def\renderwithdist#1#2#3{%
   \@ifnextchar\bgroup
   {\superfancyrender{#1}{#2}{#3}}
   {\ensuremath{\mathchoice
      {\underset{#2}{#1}\left[ #3 \right]}
      {{#1}_{#2}[ #3 ]}
      {{#1}_{#2}[ #3 ]}
      {{#1}_{#2}[ #3 ]}
     }
   }
}
\def\superfancyrender#1#2#3#4#5{
   \ensuremath{\mathchoice
      {\underset{#1}{{#1}}\left#4 #3 \right#5}
      {{#1}_{#2}#4 #3 #5}
      {{#1}_{#2}#4 #3 #5}
      {{#1}_{#2}#4 #3 #5}
   }
}
\makeatother

\newcommand{\conv}[1]{\mathrm{conv}\inparen{#1}}

\newfont{\inhead}{eufm10 scaled\magstep1}

\newcommand{\calC}{{\cal C}}
\newcommand{\calD}{{\cal D}}

\newcommand{\calF}{{\cal F}}

\newcommand{\calH}{{\cal H}}
\newcommand{\calJ}{{\cal J}}
\newcommand{\calL}{{\cal L}}

\newcommand{\calO}{{\cal O}}

\newcommand{\calT}{{\cal T}}
\newcommand{\calS}{{\cal S}}

\newcommand{\suchthat}{{\;\; : \;\;}}

\DeclareMathOperator\supp{Supp}

\DeclareMathOperator*{\argmin}{\arg\!\min}
\DeclareMathOperator*{\argmax}{\arg\!\max}

\newcommand{\ceil}[1]{\ensuremath{\left\lceil #1 \right\rceil}}

\newcommand{\bigoh}{\operatorname{O}}

\newcommand{\bigomega}{\mathop{\Omega}}

\newcommand{\inparen}[1]{\left(#1\right)}             
\newcommand{\inbraces}[1]{\left\{#1\right\}}           

\newcommand{\Caratheodory}{Carath\'eodory\xspace}
\newcommand{\Zemor}{Z\'emor\xspace}

\DeclareSymbolFont{extraup}{U}{zavm}{m}{n}
\DeclareMathSymbol{\varheart}{\mathalpha}{extraup}{86}
\DeclareMathSymbol{\vardiamond}{\mathalpha}{extraup}{87}

\def\one{\mathbf 1}

\def\tee{\mathfrak t}

\DeclareMathOperator{\PExp}{\widetilde \E}
\usepackage{tikz}

\title{List Decoding of Tanner and Expander Amplified Codes from Distance Certificates}
\author{
Fernando Granha Jeronimo\thanks{{\tt UC Berkeley}. {\tt granha@ias.edu}. Supported by the NSF grant CCF-1900460.}
\and
Shashank Srivastava\thanks{{\tt TTIC}. {\tt shashanks@ttic.edu}. Supported by the NSF grants CCF-1816372 and CCF-2326685.}
\and
Madhur Tulsiani\thanks{{\tt TTIC}. {\tt madhurt@ttic.edu}. Supported by the NSF grants CCF-1816372 and CCF-2326685.} 
}

\allowdisplaybreaks

\def\dupPE#1{
	\ProbabilityRender{\widetilde{\Esymb}^*}{#1}
}

\DeclareMathOperator{\tildecov}{\widetilde{\operatorname {Cov}}}
\DeclareMathOperator{\dupCov}{\widetilde{\operatorname {Cov}}_{\widetilde{\Esymb}^*}}
\def\li{{ \ell }}
\def\ri{{ r }}
\def\zee{{ \mathbf{Z} }}
\def\dis{{ \Delta }}
\def\tee{{ \theta }}

\def\TanC{{ \calC^{Tan} }}
\def\AELC{{ \calC^{AEL} }}

\def\indi#1{{\one \{ #1 \} }}

\def\chiTan{{ \chi}}
\def\chiAEL{{ \overline{\chi} }}

\begin{document}

\date{}

\maketitle
\draftbox
\thispagestyle{empty}

\medskip
We develop new list decoding algorithms for Tanner codes and distance-amplified codes based on
bipartite spectral expanders. 
We show that proofs exhibiting lower bounds on the minimum distance of these codes can be used as
certificates discoverable by relaxations in the Sum-of-Squares (SoS) semidefinite programming
hierarchy. 
Combining these certificates with certain entropic proxies to ensure that the solutions to the
relaxations cover the entire list, then leads to algorithms for list decoding several families of
codes up to the Johnson bound. 
We prove the following results:
\begin{itemize}
\item We show that the LDPC Tanner codes of Sipser-Spielman [IEEE Trans. Inf. Theory 1996] and Z\'{e}mor [IEEE Trans. Inf. Theory 2001] with alphabet
  size $q$, block-length $n$ and distance $\delta$, based on an expander graph with degree $d$, can be
  list-decoded up to distance $\mathcal{J}_q(\delta) - \epsilon$ in time $n^{O_{d,q}(1/\epsilon^4)}$,
  where $\mathcal{J}_q(\delta)$ denotes the Johnson bound.
\item We show that the codes obtained via the expander-based distance amplification procedure of
  Alon, Edmonds and Luby [FOCS 1995] can be list-decoded close to the Johnson bound using the SoS
  hierarchy, by reducing the list decoding problem to unique decoding of the base code. In
  particular, starting from \emph{any} base code unique-decodable up to distance $\delta$, one
  can obtain near-MDS codes with rate $R$ and distance $1-R - \epsilon$, list-decodable up to the
  Johnson bound in time $n^{O_{\epsilon, \delta}(1)}$.
\item We show that the locally testable codes of Dinur et al. [STOC 2022] with alphabet size $q$,
  block-length $n$ and distance $\delta$ based on a square Cayley complex with generator sets of
  size $d$, can be list-decoded up to distance $\mathcal{J}_q(\delta) - \epsilon$ in time
  $n^{O_{d,q}(1/\epsilon^{4})}$,  where $\mathcal{J}_q(\delta)$ denotes the Johnson bound.
\end{itemize}
%


\newpage

\pagenumbering{roman}
\tableofcontents
\clearpage

\pagenumbering{arabic}
\setcounter{page}{1}

\newpage

\section{Introduction}
\label{sec:intro}

Expander graphs have been a powerful tool for the construction of codes with several interesting
properties, and a variety of applications.
A (very) small subsample of the list of applications already includes the seminal
constructions of expander codes~\cite{SS96, Zemor01}, widely used distance amplification
constructions~\cite{ABNNR92, AEL95}, as well as recent breakthrough constructions of
$\epsilon$-balanced codes~\cite{TS17}, locally testable codes~\cite{DELLM22}, and quantum LDPC
codes~\cite{PK22, LZ22}. 
A detailed account of the rich interactions between coding theory and expander graphs, and
pseudorandom objects in general, can be found in several excellent surveys and
textbooks on these areas~\cite{GSurvey04, Vadhan12, HLW06, GRS20}. 

The combinatorial and spectral structure of codes based on expander graphs often leads to very
efficient algorithms for unique-decoding. However, obtaining list-decoding for constructions based on
expanders often require incorporating additional algebraic structure in the construction, to take
advantage of the significant machinery for list-decoding using polynomials~\cite{Guruswami:survey}.
While there are certainly important counterexamples to the above statement, such as the
expander-based codes of Guruswami and Indyk~\cite{GI03} and Ta-Shma~\cite{TS17} which
allow for list-decoding, we know of few \emph{general techniques to exploit expansion for
  list-decoding.}
In this work, we consider the question of finding techniques for list-decoding from errors,
which can work in settings where no algebraic structure may be available, such as the decoding of LDPC
codes constructed from expander graphs.

Building on the significant body of work for LP decoding of expander codes~\cite{Feldman03, FWK05}, 
we consider the question of decoding as an optimization problem, which can be approached
via convex relaxations. 
We show that stronger relaxations obtained via the Sum-of-Squares (SoS) hierarchy of semidefinite
programs, can in fact be used to obtain list-decoding algorithms for several code constructions
based on expanders. 
These hierarchies can be viewed as proof systems~\cite{FKP19}, with relaxations at a level $t$ of
the hierarchy corresponding to proofs which can be carried out by reasoning about sum-of-squares of
polynomials of degree at most $t$ in the optimization variables. The proof system corresponding to a
small number of levels of the SoS hierarchy turns out to be powerful enough
to capture the distance proofs for several expander-based codes, when the proofs rely on spectral
properties of expander graphs.
Combined with generic ``covering lemmas'' which ensure that the solutions to these relaxations
include sufficient information about all codewords in the list, these can be used to design
list-decoding algorithms for several families of codes based on expanders, up to the Johnson bound
where the list size is known to be bounded.

\vspace{-5 pt}
\subsection{Tanner Codes}
Low-Density Parity Check (LDPC) codes were introduced by a foundational work of
Gallager~\cite{Gallager62} and graph-based constructions with lower bounds on the distance (based on
girth), were obtained by Tanner~\cite{Tanner81}.
Sipser and Spielman~\cite{SS96} gave the first constructions of Tanner codes with distance bounds
based on the expansion of the graph, which also admitted linear time encoding and (unique) decoding
algorithms. 
An elegant construction based on bipartite spectral expanders,  with particularly simple
(linear-time) unique-decoding algorithms, was given by \Zemor~\cite{Zemor01}.
Variants of these constructions have led to applications~\cite{RU:book} and have also been used as
building blocks in the recent constructions of locally testable codes by Dinur \etal~\cite{DELLM22}
and quantum LDPC codes by Panteleev and Kalachev~\cite{PK22} (see also \cite{LZ22}).

There exist highly efficient algorithms for the unique-decoding of these codes from both probabilistic
and adversarial errors, based on combinatorial arguments, linear programming
relaxations~\cite{Feldman03, ADS12, FWK05} and message passing
algorithms~\cite{Guruswami:MP-survey, RU:book}. 
In the setting of \emph{erasures} where the location of the corruptions in the transmitted codeword is known,
recent work has also led to linear-time list decoding algorithms~\cite{RZWZ21, HW15}, which also
work for the more general task of list recovery in the large alphabet (high-rate) case~\cite{HW15}.
However, to the best of our knowledge, no list decoding algorithms are known in the more challenging
(and common) setting of \emph{errors} when the location of the corruptions are unknown, even though
random ensembles of LDPC codes are even known to combinatorially achieve list-decoding
capacity~\cite{MosheiffRRSW19}, and thus have bounded list sizes up to optimal error radii.

We show that relaxations obtained via the SoS hierarchy can be used list-decode \Zemor's
construction of Tanner codes~\cite{Zemor01}, up to the Johnson bound (which is an error-radius where
list sizes are always known to be bounded). 
Our proof technique can also be extended to work for other constructions of Tanner codes where the
proof for the distance of the code is based on spectral arguments, but is easiest to illustrate in
the context of \Zemor's construction. 
We briefly recall the construction before describing our result.

Given a bipartite $d$-regular graph $G=(L,R,E)$ with $\abs{L}=\abs{R}=n$, \Zemor's construction
yields a code with block-length $m = \abs{E} = nd$. 
Given an alphabet size $q$, the code consists of all edge-labelings $f \in [q]^m$, such that the
labels in the neighborhood of every vertex
\footnote{One can also consider variants where the base code $\calC_0$ is different for different
  vertices, but this does not make a difference for our purposes.}, 
belong to a ``base code'' $\calC_0 \subseteq [q]^d$.
When the base code has (fractional) distance $\delta_0$ and $G$ has (normalized) second singular
value at most $\lambda$ for the biadjacency matrix, the distance of the code is known to be at least $\delta = \delta_0 \cdot
(\delta_0 - \lambda)$. 
The Johnson bound for distance $\delta$ and alphabet size $q$ is defined as
$\calJ_q(\delta) \defeq (1 - \nfrac{1}{q}) \cdot \inparen{ 1 - \inparen{1 - \nfrac{q \cdot
      \delta}{(q-1)}}^{1/2}}$, and is always greater than the unique-decoding radius $\delta/2$.
\ifnum\confversion=1
We prove the following (see full version for a formal statement).
\else
We prove the following.
\fi

\begin{theorem}\ifnum\confversion=0[Informal version of \cref{thm:tanner-decoding}]\fi
\ Given a Tanner code $\calC$ as above and $\eps > 0$, there is a deterministic algorithm based on
$q^{O(d)}/\eps^4$ levels of the SoS hierarchy, which given an arbitrary $g \in [q]^m$, runs in time
$n^{q^{O(d)}/\eps^4}$, and recovers the list of codewords within distance
$\calJ_q(\delta) - \eps$ ~of $g$. 
\end{theorem}
Note that one can think of $q, d$ in \Zemor's construction and $\eps$ above as constants (in fact
$d$ is required to be constant for LDPC codes), in which case the above running time is polynomial in
$n$. 
Of course, these running times are no match for the linear-time unique decoding, and erasure
list-decoding, algorithms available for these codes. 
However, we view the techniques used in the proof of the above algorithm as a first step towards
identifying the right structures, and designing truly efficient algorithms, to take advantage of
expansion for list-decoding of LDPC codes (as has proved to be the case for several SoS-based algorithms in the past). 

Our techniques also extend to yield a similar statement for the recent construction of locally
testable codes by Dinur \etal~\cite{DELLM22}, which are Tanner codes on a different structure
called a ``square Cayley complex''.
These are constructed using a group $H$ and two generator sets $A,B \subseteq H$, with each
generator set individually defining an expanding Cayley graph on $H$ (with second eigenvalue bounded
by $\lambda$).  The sizes of the generator
sets (equal to the graph degree) are taken as constant, say $\abs{A} = \abs{B} =  d$. The
construction relies on base codes $\calC_A, \calC_B \subseteq [q]^d$, with distances, say $\delta_A$
and $\delta_B$, and is known to have distance at least $\delta = \delta_A \cdot \delta_B \cdot
\inparen{\max\{\delta_A, \delta_B\} - \lambda}$.
\begin{theorem}\ifnum\confversion=0[Informal version of \cref{thm:square-decoding}]\fi
\ Given a code  $\calC^{SCC}$ with block length $m$ and alphabet $[q]$, 
supported on a square Cayley complex as described above, and $\eps > 0$, there is a deterministic
algorithm based on $q^{O(d^2)}/\eps^{O(1)}$ levels of the SoS hierarchy, which given an arbitrary $g \in [q]^m$, 
recovers the list of codewords within distance $\calJ_q(\delta) - \eps$ ~of $g$. 
\end{theorem}
\vspace{-5 pt}
\subsection{Distance Amplified Codes}
The proofs for the distance of the above Tanner codes, are also very similar to the ones used for
analyzing the distance amplification procedure of Alon, Edmonds, and Luby~\cite{AEL95} (AEL), based on
expander graphs.
While there several variants of this construction discussed in the literature, we 
will discuss a version of the AEL construction~\cite{Kopparty} which is particularly close to the
Tanner code construction of \Zemor.
Given a $d$-regular bipartite graph $G=(L,R,E)$ with second singular value $\lambda$, an ``outer" code $\calC_1 \subseteq [q_1]^n$ with distance $\delta_1$, and an ``inner" code $\calC_0 \subseteq [q]^d$ with $\abs{\calC_0} = q_1$, the AEL procedure constructs a new code $\calC^{AEL} \subseteq [q^d]^n$ with distance at least $\delta_0 - \frac{\lambda}{\delta_1}$. 
Thus, it yields constructions with arbitrarily large block lengths that inherit the parameters of the small inner code. 

The AEL procedure has been used as an important ingredient for obtaining optimal rate-distance
tradeoffs in several constructions, such as the capacity-achieving list-decodable codes by Guruswami
and Rudra~\cite{GuruswamiR06}. The amplification is achieved via a simple redistribution of symbols using the expander, and the construction also preserves several interesting local properties of the outer code, such as the property of being LDPC, or locally testable, or locally correctable~\cite{KMRZS17, GKORZS18}. 
We refer the reader to the discussion in \cite{KMRZS17} for an excellent account of the applications and properties of the AEL construction.

The AEL procedure has been used to construct several list-decodable codes, including some of the
results cited above, and a quantum analogue of the construction was also used recently by Bergamaschi \etal~\cite{BGG22} to obtain quantum codes meeting the Singleton bound (via quantum list decoding). 
However, for the resulting code $\calC^{AEL}$ to be list-decodable, one often needs to assume stronger properties such as list-recovery, or some algebraic structure (or both) for the outer code $\calC_1$. 
Since these stronger properties may not always be available (for example, when one wants to preserve some local properties for $C_1$ like being LDPC), we again consider the question of finding techniques which can allow for list-decoding $\calC^{AEL}$ for expanding graphs $G$, without relying on additional structure from $\calC_1$.

We show that relaxations based on the SoS hierarchy, can be used to list-decode the distance-amplified code $\calC^{AEL}$, even when the outer code is only assumed to be \emph{unique-decodable}. In particular, we prove the following result\ifnum\confversion=1 \ (a formal statement appears in full version)\fi:
\begin{theorem}\ifnum\confversion=0[Informal version of \cref{thm:list_decoding_ael}]\fi
\ Let $\calC^{AEL}$ be a distance-amplified code as above, with the outer code $\calC_1$ taken to be unique-decodable from radius $\delta_{dec}$ in  time $T(n)$, and let $\delta = \delta_0 - \frac{\lambda}{\delta_{dec}}$. Then, for every $\eps > 0$,
there is a deterministic algorithm based on $q^{O(d)}/\eps^4$ levels of the SoS hierarchy, which
given an arbitrary $g \in [q^d]^n$, runs in time
$n^{q^{O(d)}/\eps^4} + O(T(n))$, and recovers the list of codewords within distance
$\calJ_{q^d}(\delta) - \eps$ ~of $g$.
\end{theorem}
\vspace{-5pt}
We note that the decoding radius for the above algorithm is $\calJ\inparen{\delta_0 - \frac{\lambda}{\delta_{dec}}}$ instead of $\calJ\inparen{\delta_0 - \frac{\lambda}{\delta_{1}}}$, which would be the Johnson bound for the true distance of the code $\calC^{AEL}$. 
However, in applications of AEL, one often chooses parameters so that the distance of code is about
$\delta_0$, and the effect of the second term is minimized by choosing a small $\lambda$. 
When $\calC_1$ is known to unique-decodable up to a smaller radius $\delta_{dec}$, one can still obtain list-decodable codes up to (nearly) the Johnson bound by choosing $G$ to be a sufficiently good expander (with small $\lambda$).

\vspace{-5 pt}
\subsection{Our techniques}
As mentioned earlier, our techniques are based on using the Sum-of-Squares hierarchy of convex
relaxations for an optimization problem related to the decoding problem. 
In unique-decoding algorithms based on the LP relaxations, the optimization objective is to find the
closest codeword to a given received word, and the correctness of the decoding procedure often
relies on the LP being integral for an appropriate range of parameters.
In contrast, algorithms for list-decoding actually need to ensure that the solution to the convex
relaxation has sufficient information about all codewords in the list, and so it is important that
the solution is \emph{not} integral but rather a ``maximally-convex'' combination, covering all of
the list elements. 
This can be ensured by statements which we call ``covering lemmas'', which are discussed in more detail in \ifnum\confversion=1 \cref{sec:overview}. \else \cref{sec:covering}. \fi
The proofs for the covering lemmas are based on the techniques from~\cite{AJQST20}, where these were
used for the list-decoding of direct-sum codes.

A second key component of our proof, which makes the SoS hierarchy particularly appealing to work
with, is for the relaxation to be able to capture global properties of the code, such as the
distance. 
While local properties of the code, such as the structures of the base/inner codes are enforced
through explicit constraints included in the relaxation, the global property of distance is a
nontrivial consequence of these constraints.
However, the \emph{proofs} of these distance properties are spectral in nature, for the codes we consider
here, which makes them discoverable by the SoS hierarchy.

For example, the proofs rely on statements such as the expander mixing lemma, which can viewed as a
consequence of statements like $\ip{f}{A f} \leq \lambda \cdot \norm{f}^2$ when $\ip{f}{1} = 0$, and
$A$ is the (normalized) adjacency matrix of a graph with second eigenvalue at most $\lambda$. 
Taking $\Pi$ to be the projector to the space orthogonal to the all-ones vector, we can re-write the
above inequality as $\ip{f}{(\lambda \cdot \Pi - \Pi A \Pi)f} \geq 0$. 
However, note that the matrix $\lambda \cdot \Pi - \Pi A \Pi$ is actually a positive semidefinite
matrix, which means that the expression $\ip{f}{(\lambda \cdot \Pi - \Pi A \Pi)f}$ is a
\emph{sum-of-squares} of linear forms in the entries of $f$. 
The SoS hierarchy can be viewed as a proof system, where a solution to the level-$2t$ relaxation
can be seen as satisfying all inequalities which can be derived using sum-of-squares of polynomials
of degree at most $t$. 
We can show that this means that the solutions (after some modification) satisfy some codeword-like
properties, using which it is possible to appeal to a \emph{unique} decoding algorithm to recover
one element from the list from one such ``codeword-like'' SoS solution.
\ifnum\confversion=1
See the full version for a formal statement about these codeword-like ``distance certificates'' for SoS solutions.
\else
These codeword-like ``distance certificates'' for SoS solutions are developed in \cref{sec:distance}.
\fi

Broadly speaking, our techniques can be seen as part of the ``Proofs to Algorithms'' paradigm based
on the Sum-of-Squares method~\cite{FKP19}. 
Our covering lemmas for SoS relaxations yield  a generic framework for converting SoS proofs of
distance for \emph{any code}, to list decoding algorithms which work up to the Johnson bound. 
Examples for such results also include the earlier results by Richelson and Roy~\cite{RR22} for list
decoding Ta-Shma's codes up to  Johnson bound.

We also note that list-decoding algorithms often need to rely on algebraic structure, and are thus particularly
well suited to work with large alphabets (fields). One then obtains algorithms for small-alphabet
codes via techniques such as concatenation and list recovery.
On the other hand, the techniques based on convex relaxations discussed above seem to work well
directly over small alphabets.
\vspace{-5 pt}
\subsection{Related work}
In terms of techniques, the works most directly related to ours are those using similar SoS
relaxations for list-decoding of Ta-Shma's codes~\cite{AJQST20, JQST20, RR22}. 
In particular, the proofs of the covering lemmas follow the approach of Alev \etal~\cite{AJQST20},
and idea of viewing the proof of distance as implementable in the SoS hierarchy was also used by
 Richelson and Roy~\cite{RR22}. 
A precursor to much of this research on list-decoding, 
is the result of Dinur \etal~\cite{DHKLNTS19}, which suggested
the approach of using semidefinite programming and expansion for list-decoding of codes obtained via
an earlier distance amplification procedure of Alon \etal~\cite{ABNNR92}, which can be seen as a
special case of the AEL distance amplification.

Another important work, related to the use of SoS hierarchy for decoding LDPC codes, is the
\emph{lower bound} of Ghazi and Lee~\cite{GL18} for using the SoS hierarchy to decode random LDPC
codes. 
However, the lower bound shows that the relaxation for finding the optimal (closest) codeword may
have value much better than the true optimal codeword, when the decoding radius is larger than that
of LP decoding, thus showing that the SoS relaxations may not be integral. 
On the other hand, the relaxations we use do not optimize for the closest codeword, but rather go
through covering lemmas.
A recent work of Chen \etal~\cite{CCLO22} also shows significantly improved distance bounds, and
improved unique-decoding bounds for the expander codes of Sipser and Spielman~\cite{SS96}. 
Our results do not apply for the codes considered in their work in a black box fashion, since the analysis is based on
lossless \emph{vertex expansion}, for which we do not always know of spectral certificates.

Our work can also be seen as obtaining ``sparse'' analogues of the results of Gopalan, Guruswami, and
Raghavendra~\cite{GGR09} for the list-decoding of tensor and interleaved codes, which can be viewed
as replacing the bipartite expanders in \Zemor and AEL constructions respectively, by a complete bipartite graph.


\section{Preliminaries and Notation}
\label{sec:prelims}
For a bipartite graph $G=(L,R,E)$ where $L$ is the set of left vertices, and $R$ is the set of right
vertices, we index the left set by $\li$ and the right set by $\ri$. For a vertex $\li \in L$, 
we denote the set of edges incident to it by $N_L(\li)$ (left neighborhood), and the set of edges
incident to  $\ri\in R$ is denoted by $N_R(\ri)$(right neighborhood). 
We use $\li \sim \ri$ to denote that the vertex $\li \in L$ is adjacent to the vertex $\ri \in R$, that is, $(\li,\ri)\in E$. 

Fix an arbitrary ordering of the edges. Then there are bijections between the sets $E$, $L \times
[d]$, and $R \times [d]$, 
given by taking $(\li,i)$ to be the $i^{th}$ edge incident on $\li$, and similarly for $R \times [d]$.
Henceforth, we will implicitly assume such an ordering of the edges is fixed, and use the resulting bijections.

\begin{definition}
	Let $[q]$ be a finite alphabet and let $f,g\in [q]^n$. Then the (fractional) distance
        between $f,g$ is defined as \[ \dis(f,g) = \Ex{i\in [n]}{ \indi{f_i \neq g_i}} \mper \]
\end{definition}

\begin{definition}[Code, distance and rate]
	A code $\calC$ of block length $n$, distance $\delta$ and rate $\rho$ over the alphabet size $q$ is a set $\calC \subseteq [q]^n$ with the following properties
	\begin{enumerate}[(i)]
		\item $\rho = \frac{\log_{q} |\calC|}{n}$
		\item $\delta = \min_{h_1,h_2\in \calC} \dis(h_1,h_2)$
	\end{enumerate}
	Such codes are succinctly represented as $[n,\delta,\rho]_q$.
	We say $\calC$ is a linear code if $[q]$ can be identified with a finite field and $\calC$ is a linear subspace of $[q]^n$.
\end{definition}

\begin{definition}[List of codewords]
	For any $g\in [q]^n$, the list of codewords in $\calC$ that are at a distance less than
        $\delta$ from $g$ is denoted by $\calL(g,\delta)$. That is $\calL(g,\delta) = \inbraces{h\in \calC \suchthat \dis(h,g) < \delta }$.
\end{definition}

\subsection{Expander graphs}
\begin{definition}[$(n,d,\lambda)$-expander]
A $d$-regular bipartite graph $G(L,R,E)$ with $|L|=|R|=n$ is said to be an $(n,d,\lambda)$-expander if
\[
	\sigma_2(A_G) \leq \lambda \cdot d
\]
where $A_G$ is the $L\times R$ biadjacency matrix, and $\sigma_2(A_G)$ is its second largest singular value.
\end{definition}

Infinite families of $(n,d,\lambda)$-expanders, with growing $n$ as $d$ and $\lambda$ are constant, can be derived based on double covers of Ramanujan graphs of \cite{LPS88} as long as $\lambda \geq \frac{2\sqrt{d-1}}{d}$.

\begin{lemma}[Expander Mixing Lemma]
	Given an $(n,d,\lambda)$-expander $G=(L,R,E)$ and functions $f: L \rightarrow \R$ and $g: R \rightarrow \R$, the following well known property is a simple consequence of definition of $(n,d,\lambda)$-expanders:
	\ifnum\confversion=1
	\small
	\[
		\abs{\Ex{\li \sim \ri}{ f(\li) \cdot g(\ri)} - \Ex{\li}{f(\li)} \Ex{\ri}{g(\ri)}}
                \leq \lambda \cdot \norm{f}_2 
                \norm{g}_2
	\]
	\normalsize
	\else
	\[
		\abs{\Ex{\li \sim \ri}{ f(\li) \cdot g(\ri)} - \Ex{\li}{f(\li)} \cdot \Ex{\ri}{g(\ri)}}
                \leq \lambda \cdot \inparen{\Ex{\li}{f^2(\li)}}^{1/2} \cdot
                \inparen{\Ex{\ri}{g^2(\ri)}}^{1/2} \mper
	\]
	\fi
\end{lemma}

\subsection{Graph based codes}
\paragraph{Expander codes.}
We recap the construction from~\cite{Zemor01}, building on ideas from \cite{Tanner81, SS96},
which constructs an infinite family of good codes starting from any good (inner) linear code over a
small fixed block length of rate larger than $1/2$. The code $\calC_0$ is also referred to as the
base code for $\calC^{Tan}$.
\begin{definition}
	Given an inner linear code $\calC_0$ on alphabet $[q]$ and block length $d$, and a $d$-regular graph $G(V,E)$, we define the Tanner code $\TanC$ as
	\[
		\TanC = \{ h: E\rightarrow [q] \suchthat h|_{N(v)} \in \calC_0, \forall v \in V\}
	\]
	where $h|_S$ for $S\subseteq E$ denotes the restriction of $h$ to the set of coordinates indexed by $S$.
	In this work, we will only use Tanner codes defined on bipartite graphs. 
\end{definition}
By using an infinite family of graphs with constant degree $d$, we get an infinite family of codes based on an inner code of block length $d$.

\begin{theorem}[\cite{SS96}, \cite{Zemor01}]
	Let the distance and rate of the inner code $\calC_0$ be $\delta_0$ and $\rho_0$ respectively, and those of associated $\TanC$ be $\delta$ and $\rho$. If $G$ is an $(n,d,\lambda)$-expander, then $\delta \geq \delta_0\cdot (\delta_0-\lambda)$ and $\rho\geq 2\rho_0-1$.
\end{theorem}
The codes above were shown to be linear time decodable up to their unique decoding radius (for the designed distance) in \cite{Zemor01} and \cite{SkaRoth03}.
\paragraph{Alon-Edmonds-Luby distance amplification}\label{sec:AEL_prelims}

The following distance amplification scheme was introduced in \cite{ABNNR92}, \cite{AEL95} and used by \cite{GI05} to design linear time unique decodable near-MDS codes.

\begin{definition}[Concatenated codes]
	Fix an $(n,d,\lambda)$-expander $G(L,R,E)$. Let $\calC_{1}$ be an $[n,\delta_{1},r_{1}]_{q_1}$ code and let $\calC_{0}$ be a $[d,\delta_{0},r_{0}]_{q_0}$ code with $q_1 = |\calC_0|$. 
	
	We define the concatenation of $f\in [q_1]^L$ with $\calC_{0}: [q_1]\rightarrow [q_0]^d$ as $f^*: E \rightarrow [q_0]$, defined as
	\[
		f^{*}_{\calC_{0}}(e) ~=~ \calC_{0}(f(\li))(j)
	\]
	where $\li$ is the left endpoint of edge $e$ and $e$ is the $j^{th}$ edge incident on $\li$.
	The concatenated code $\calC^{*}_{\calC_{0}}(\calC_1)$ of block length $n\cdot d$ and alphabet $[q_0]$ is defined to be,
	\[
		\calC^{*}_{\calC_{0}}(\calC_{1}) ~=~ \{ f^{*}_{\calC_{0}} \suchthat f \in \calC_{1}\}
	\]
	When clear from context, we will omit $\calC_{0},\calC_1$ in the above notation to call the concatenated code $\calC^*$.
\end{definition}

\begin{claim}
	$\dis(f^{*}_{\calC_{0}},g^{*}_{\calC_{0}}) ~\geq~ \delta_{0}\cdot \dis(f,g)$, which also
        implies $\Delta(\calC^{*}_{\calC_{0}}(\calC_{1})) \geq \delta_{0} \cdot \delta_{1}$.
\end{claim}

\begin{definition}[AEL Codes]
	Fix an $(n,d,\lambda)$-expander $G(L,R,E)$. Let $\calC_{1}$ be an $[n,\delta_{1},r_{1}]_{q_1}$ code and let $\calC_{0}$ be a $[d,\delta_{0},r_{0}]_{q_0}$ code with $q_1 = |\calC_0|$. 
We define the AEL-encoding $f^{AEL}_{\calC_{0}}: R \rightarrow [q_0]^d$ as
	\[
		f^{AEL}_{\calC_{0}} (\ri) ~=~ \left( f^{*}_{\calC_{0}}(e_1),f^{*}_{\calC_{0}}(e_2),\cdots,f^{*}_{\calC_{0}}(e_d) \right)
	\]
	where $e_1,e_2,\cdots,e_d$ are the $d$ edges incident on $\ri$.
	The AEL code $\calC^{AEL}_{\calC_{0}}(\calC_{1}) \subseteq [q_0^d]^{n}$ is defined as 
	\[
		\calC^{AEL}_{\calC_{0}}(\calC_{1}) ~=~ \{ f^{AEL}_{\calC_{0}} \suchthat f \in \calC_{1}\}
	\]
	When clear from context, we will omit $\calC_{0},\calC_1$ in the above notation to call the AEL code $\AELC$.
\end{definition}
Alon, Edmonds and Luby, proved the following result, which shows that the construction can be used
to amplify the distance to $\delta_0$, by choosing $\lambda$ sufficiently small.
\begin{theorem}[\cite{AEL95}]\label{thm:ael_distance}
	$\dis(f^{AEL}_{\calC_{0}},g^{AEL}_{\calC_{0}}) ~\geq~ \delta_{0}-
        \frac{\lambda}{\dis(f,g)}$, which also implies $\Delta(\calC^{AEL}_{\calC_{0}}(\calC_{1}))
        ~\geq~ \delta_{0} - \frac{\lambda}{\delta_{1}}$.
\end{theorem}

A codeword $f$ of $\AELC$ technically belongs to the space $[q_0^d]^R$. However, in this paper, we choose to identify codewords of $\AELC$ as belonging to $[q_0]^E$. It is clear that the two spaces are in bijection with each other and in fact, one can just \emph{fold} or \emph{unfold} the symbols to move from one space to the other. Choosing $f$ to be in $[q_0]^E$ allows us to talk about $f$ viewed from left vertex set $L$ or right vertex set $R$ simply by changing the distance functions. Let $f,g \in \AELC$, then we define the following three distance functions:
\begin{align*}
	\dis^L(f,g) &\defeq \Ex{\li\in L}{ \indi{ f_{N_L(\li)} \neq g_{N_L(\li)}}} \\
	\dis(f,g) &\defeq \Ex{e\in E}{\indi{ f_e \neq g_e}} = \Ex{\li \in L}{\dis(f_{N_L(\li)},g_{N_L(\li)})} = \Ex{\ri \in R}{\dis(f_{N_R(\ri)},g_{N_R(\ri)})}\\
	\dis^R(f,g) &\defeq \Ex{\ri\in R}{ \indi{ f_{N_R(\ri)} \neq g_{N_R(\ri)}}} 
\end{align*}

\cref{thm:ael_distance} can then be stated in a simpler form as
\begin{theorem}[Restatement of \cref{thm:ael_distance}]
	$\dis^R(f,g) \geq \delta_0 - \frac{\lambda}{\dis^L(f,g)}$.
\end{theorem}
\subsection{Sum-of-Squares hierarchy}
The sum-of-squares hierarchy of semidefinite programs (SDPs) provides a family of increasingly
powerful convex relaxations for several optimization problems. 
Each ``level" $t$ of the hierarchy is parametrized by a set of constraints corresponding to
polynomials of degree at most $t$ in the optimization variables. While the relaxations in the
hierarchy can be viewed as  semidefinite programs of size $n^{O(t)}$ \cite{BS14, Laurent09}, 
it is often convenient to view the solution as a linear operator, called the ``pseudoexpectation" operator.
\vspace{-5 pt}
\paragraph{Pseudoexpectations}
Let $t$ be an positive even integer and fix an alphabet $[q]$. An SoS solution of degree $t$, or a pseudoexpectation of SoS-degree $t$, over the variables $\zee = \{Z_{i,j}\}_{i\in[m],j\in[q]}$ is represented by a linear operator $ \tildeEx{\cdot}: \R[\zee]^{\leq t} \rightarrow \R$ such that:
\vspace{-5 pt}
\begin{enumerate}[(i)]
    \item $\tildeEx{1} = 1$.
    \item $\tildeEx{p^2} \geq 0$ if $p$ is a polynomial in $\zee = \{Z_{i,j}\}_{i\in [m],j\in [q]}$ of degree $\leq t/2$.
\end{enumerate}
\vspace{-5 pt}
 Note that linearity implies $\tildeEx{p_1} + \tildeEx{p_2} = \tildeEx{p_1+p_2}$ and $\tildeEx{c\cdot
  p_1} = c \cdot \tildeEx{p_1}$ for $c\in \R$, for $p_1, p_2 \in \R[\zee]^{\leq t}$.
This also allows for a succinct representation of $\tildeEx{\cdot}$ using any basis for $\R[\zee]^{\leq t}$.

The set of all pseudoexpectations should be seen as a relaxation for the set of all possible
(distributions over) assignments to $m$ variables in alphabet $[q]$.
Indeed, any assignment $f: [m] \rightarrow [q]$, can be seen as a pseudoexpectation 
which assigns the value $1$ to a
monomial consistent with $f$ and $0$ otherwise. This can be extended via linearity to all
polynomials, and then by convexity of the constraints to all distributions over assignments.
However, the reverse is not true when $t < m$, and there can be degree-$t$ pseudoexpectations which do
not correspond to any genuine distribution.

It is possible to optimize over the set of degree-$t$ pseudoexpectations in time $m^{O(t)}$ via SDPs
(under certain conditions on the bit-complexity of solutions~\cite{OD16, RW17:sos}).
We next define what it means for pseudoexpectations to satisfy some problem-specific constraints.

\begin{definition}[Constrained Pseudoexpectations]\label{def:constraints_on_sos}
Let $\calS = \inbraces{f_1 = 0, \ldots, f_m = 0, g_1 \geq 0, \ldots, g_r \geq 0}$ be a system of
polynomial constraints, with each polynomial in $\calS$ of degree at most $t$. We say $\tildeEx{\cdot}$ is a pseudoexpectation operator respecting $\calS$, if in addition to the above conditions, it also satisfies
	\begin{enumerate}
	\item $\tildeEx{p \cdot f_i} = 0$,  $\forall i \in [m]$ and $\forall p$ such that $\deg(p \cdot f_i) \leq t$.
	\item $\tildeEx{p^2 \cdot \prod_{i \in S} g_i} \geq 0$, $\forall S \subseteq [r]$ and $\forall p$ such that $\deg(p^2 \cdot \prod_{i \in S} g_i) \leq t$.
	\end{enumerate}
\end{definition}
\paragraph{Local constraints and local functions.}
Any constraint that involves at most $k$ variables from $\zee$, with $k\leq t$, can be written as a degree-$k$ polynomial, and such constraints may be enforced into the SoS solution.
In particular, we will always consider the following canonical constraints on the variables $\zee$.
\ifnum\confversion=1
\begin{align*}
&Z_{i,j}^2 = Z_{i,j},\ \forall i\in[m],j\in[q] \\
\text{and} \quad &\sum_j Z_{i,j} = 1,\ \forall i\in[m] \mper
\end{align*}
\else
\[
Z_{i,j}^2 = Z_{i,j},\ \forall i\in[m],j\in[q] 
\quad \text{and} \quad 
\sum_j Z_{i,j} = 1,\ \forall i\in[m] \mper
\]
\fi
We will also consider additional constraints and corresponding polynomials, defined by ``local" functions. For any $f\in [q]^m$ and $M\sub [m]$, we use $f_M$ to denote the restriction $f|_M$, and $f_i$ to denote $f_{\{i\}}$ for convenience.
\begin{definition}[$k$-local function]
	A function $\mu: [q]^m \rightarrow \R$ is called $k$-local if there is a set $M\subseteq [m]$ of size $k$ such that $\mu(f)$ only depends on $\inbraces{f(i)}_{i\in M}$, or equivalently, $\mu(f)$ only depends on $f|_M$.
	
	If $\mu$ is $k$-local, we abuse notation to also use $\mu: [q]^M \rightarrow \R$ with $\mu(\alpha) = \mu(f)$ for any $f$ such that $f|_M=\alpha$. It will be clear from the input to the function $\mu$ whether we are using $\mu$ as a function on $[q]^m$ or $[q]^M$.
\end{definition}

Let $\mu:[q]^m\rightarrow \R$ be a $k$-local function that depends on coordinates $M\subseteq [m]$ with $|M|=k$. Then $\mu$ can be written as a degree-$k$ polynomial $p_{\mu}$ in $\zee$:
\[
	p_{\mu}(\zee) = \sum_{\alpha \in [q]^M} \inparen{\mu(\alpha) \cdot\prod_{i\in M} Z_{i,\alpha_i}}
\]

With some abuse of notation, we let $\mu(\zee)$ denote $p_{\mu}(\zee)$. We will use such $k$-local
functions inside $\tildeEx{\cdot}$ freely without worrying about their polynomial
representation. For example, $\tildeEx{ \indi{\zee_{i} \neq j}}$ denotes $\tildeEx{ 1- Z_{i,j}}$. 
The notion of $k$-local functions can also be extended from real valued functions to vector valued functions in a straightforward way.

\begin{definition}[vector-valued local functions]
A function $\mu: [q]^m \rightarrow \R^N$ is $k$-local if the $N$ real valued functions corresponding to the $N$ coordinates are all $k$-local. Note that these different coordinate functions may depend on different sets of variables, as long as the number is at most $k$ for each of the functions.
\end{definition}
\vspace{-5 pt}
\paragraph{Local distribution view of SoS}

It will be convenient to use a shorthand for the function $\indi{\zee_S = \alpha}$, and we will use $\zee_{S,\alpha}$. Likewise, we use $\zee_{i,j}$ as a shorthand for the function $\indi{\zee_i = j}$. That is, henceforth,
\ifnum\confversion=1
\begin{align*}
	&\tildeEx{\zee_{S,\alpha}} = \tildeEx{\indi{\zee_S = \alpha}} = \tildeEx{ \prod_{s\in S}
                                    Z_{s,\alpha_s}}  \\
	\text{and } \quad & \tildeEx{\zee_{i,j}} = \tildeEx{\indi{\zee_i = j}} = \tildeEx{ Z_{i,j}}.
\end{align*}
\else
\begin{align*}
	\tildeEx{\zee_{S,\alpha}} ~=~ \tildeEx{\indi{\zee_S = \alpha}} = \tildeEx{ \prod_{s\in S}
                                    Z_{s,\alpha_s}} 
\qquad \text{and} \qquad
	\tildeEx{\zee_{i,j}} ~=~ \tildeEx{\indi{\zee_i = j}} = \tildeEx{ Z_{i,j}}
\end{align*}
\fi

Note that for any $S \subseteq [m]$ with $\abs{S} = k \leq t/2$,
\ifnum\confversion=1
\begin{gather*}
	\sum_{ \alpha \in [q]^{k}} \tildeEx{\zee_{S,\alpha}} = 
 \tildeEx{ \prod_{s\in S} \inparen{ \sum_{j\in [q]} Z_{s,j}} } = 1
\\
	\tildeEx{\zee_{S,\alpha}} = \tildeEx{ \prod_{s\in S} Z_{s,\alpha_s}} = \tildeEx{ \prod_{s\in
            S} Z^2_{s,\alpha_s}} \geq 0 \mper
\end{gather*}
\else
\[
	\sum_{ \alpha \in [q]^{k}} \tildeEx{\zee_{S,\alpha}} = 
 \tildeEx{ \prod_{s\in S} \inparen{ \sum_{j\in [q]} Z_{s,j}} } = 1
\qquad \text{and} \qquad
	\tildeEx{\zee_{S,\alpha}} = \tildeEx{ \prod_{s\in S} Z_{s,\alpha_s}} = \tildeEx{ \prod_{s\in
            S} Z^2_{s,\alpha_s}} \geq 0 \mper
\]
\fi
Thus, the values $\inbraces{\tildeEx{\zee_{S, \alpha}}}_{\alpha\in [q]^S}$ define a distribution
over $[q]^k$, referred to as the local  distribution for $\zee_S$.

Let $\mu: [q]^m \rightarrow\R$ be a $k$-local function for $k\leq t/2$, depending on $M \subseteq
[m]$. Then, $\tildeEx{\mu(\zee)}$ can be seen as the expected value of the function $\mu$ under the
local distribution for $M$, since
\ifnum\confversion=1
\begin{align*}
	\tildeEx{\mu(\zee)} 
~=~& \tildeEx{\sum_{\alpha\in [q]^M} \inparen{\mu(\alpha) \cdot\prod_{i\in M} Z_{i,\alpha_i}}}\\
~=~& \sum_{\alpha\in [q]^M} \mu(\alpha) \cdot \tildeEx{\prod_{i\in M} Z_{i,\alpha_i}}\\
~=~& \sum_{\alpha\in [q]^M} \mu(\alpha) \cdot \tildeEx{\zee_{M,\alpha}} \mper
\end{align*}
\else
\begin{align*}
	\tildeEx{\mu(\zee)} 
~=~ \tildeEx{\sum_{\alpha\in [q]^M} \inparen{\mu(\alpha) \cdot\prod_{i\in M} Z_{i,\alpha_i}}}
~=~ \sum_{\alpha\in [q]^M} \mu(\alpha) \cdot \tildeEx{\prod_{i\in M} Z_{i,\alpha_i}}
~=~ \sum_{\alpha\in [q]^M} \mu(\alpha) \cdot \tildeEx{\zee_{M,\alpha}} \mper
\end{align*}
\fi

\begin{claim}
	Let $\tildeEx{\cdot}$ be a degree-$t$ pseudoexpectation. For $k \leq t/2$, let $\mu_1,\mu_2$
        be two $k$-local functions on $[q]^m$, depending on the same set of coordinates $M$, and
        $\mu_1(\alpha) \leq \mu_2(\alpha) ~~\forall \alpha \in [q]^M$. Then $\tildeEx{\mu_1(\zee)} \leq \tildeEx{\mu_2(\zee)}$.
\end{claim}

\begin{proof}
Let $\calD_M$ be the local distribution induced by $\tildeEx{\cdot}$ for $\zee_M$. Then
$\tildeEx{\mu_1(\zee)} = \Ex{\alpha \sim \calD_M}{\mu_1(\alpha)}$, and $\tildeEx{\mu_2(\zee)} =
\Ex{\alpha\sim \calD_M}{\mu_2(\alpha)}$, which implies $\tildeEx{\mu_1(\zee)} \leq \tildeEx{\mu_2(\zee)}$.
\end{proof}
The previous claim allows us to replace any local function inside $\tildeEx{\cdot}$ by another local function that dominates it. We will make extensive use of this fact.
\vspace{-5 pt}
\paragraph{Covariance for SoS solutions}
Given two sets $S,T \sub [m]$ with $|S|,|T|\leq k/4$, we can define the covariance between indicator random variables of $\zee_S$ and $\zee_T$ taking values $\alpha$ and $\beta$ respectively, according to the local distribution over $S \cup T$. This is formalized in the next definition.
\begin{definition}
Let $\tildeEx{\cdot}$ be a pseudodistribution operator of SoS-degree-$t$, and $S,T$ are two sets of
size at most $t/4$, and $\alpha\in [q]^S$, $\beta\in [q]^T$, we define the pseudo-covariance and
pseudo-variance,
\ifnum\confversion=1
\small
\begin{gather*}
\tildecov(\zee_{S,\alpha},\zee_{T,\beta}) 
= \tildeEx{ \zee_{S,\alpha} \cdot \zee_{T,\beta} } - \tildeEx{\zee_{S,\alpha}} \tildeEx{\zee_{T,\beta}} \\	
\tildeVar{\zee_{S,\alpha}} ~=~ \tildecov(\zee_{S,\alpha},\zee_{S,\alpha})
\end{gather*}
\normalsize
\else
\begin{align*}
\tildecov(\zee_{S,\alpha},\zee_{T,\beta}) 
~=~ \tildeEx{ \zee_{S,\alpha} \cdot \zee_{T,\beta} } - \tildeEx{\zee_{S,\alpha}} \tildeEx{\zee_{T,\beta}}
\qquad \text{and} \qquad		
\tildeVar{\zee_{S,\alpha}} ~=~ \tildecov(\zee_{S,\alpha},\zee_{S,\alpha})
\end{align*}
\fi
The above definition is extended to pseudo-covariance and pseudo-variance for pairs of sets $S,T$, 
as the sum of absolute value of pseudo-covariance for all pairs $\alpha,\beta$ :
\ifnum\confversion=1
\begin{gather*}
\tildecov(\zee_S,\zee_T) 
~=~ \sum_{\alpha\in [q]^S \atop \beta\in [q]^T} \abs{ \tildecov(\zee_{S,\alpha},\zee_{T,\beta}) } \\
\tildeVar{\zee_S} ~=~ \sum_{\alpha\in [q]^S} \abs{ \tildeVar{\zee_{S,\alpha} } }
\end{gather*}
\else
\[
\tildecov(\zee_S,\zee_T) 
~=~ \sum_{\substack{\alpha\in [q]^S \\ \beta\in [q]^T}} \abs{ \tildecov(\zee_{S,\alpha},\zee_{T,\beta}) }
\qquad \text{and} \qquad		
\tildeVar{\zee_S} ~=~ \sum_{\alpha\in [q]^S} \abs{ \tildeVar{\zee_{S,\alpha} } }
\]
\fi
\end{definition}

We will need the fact that $\tildeVar{\zee_S}$ is bounded above by 1, since,
\ifnum\confversion=1
\begin{align*}
\tildeVar{\zee_S} 
&~=~ \sum_{\alpha} \abs{\tildeVar{\zee_{S,\alpha}}} \\
&~=~ \sum_{\alpha}\inparen{
          \tildeEx{\zee_{S,\alpha}^2} - \tildeEx{\zee_{S,\alpha}}^2} \\
&~\leq~ \sum_{\alpha} \tildeEx{\zee_{S,\alpha}^2} \\
&~=~ \sum_{\alpha} \tildeEx{\zee_{S,\alpha}} 
~=~ 1
\end{align*}
\else
\[
\tildeVar{\zee_S} 
~=~ \sum_{\alpha} \abs{\tildeVar{\zee_{S,\alpha}}} 
~=~ \sum_{\alpha}\inparen{
          \tildeEx{\zee_{S,\alpha}^2} - \tildeEx{\zee_{S,\alpha}}^2} 
~\leq~ \sum_{\alpha} \tildeEx{\zee_{S,\alpha}^2} 
~=~ \sum_{\alpha} \tildeEx{\zee_{S,\alpha}} 
~=~ 1
\]
\fi

\vspace{-5 pt}
\paragraph{Conditioning SoS solutions.}
We will also make use of conditioned pseudoexpectation operators, which may be defined in a way
similar to usual conditioning for true expectation operators, as long as the event we condition on
is local. 
The conditioned SoS solution is of a smaller degree, but continues to respect the constraints that original solution respects.

\begin{definition} Let $F \subseteq [q]^m$ be subset (to be thought of as an event) such that $\one_F:[q]^m \rightarrow \{0,1\}$ is a $k$-local function. Then for every $t>2k$, we can condition a pseudoexpectation operator of SoS-degree $t$ on $F$ to obtain a new conditioned pseudoexpectation operator $\condPE{\cdot}{E}$ of SoS-degree $t-2k$, as long as $\tildeEx{\one^2_F(\zee)}>0$. The conditioned SoS solution is given by
\[
	\condPE{ p(\zee)}{F(\zee) } \defeq \frac{\tildeEx{p(\zee) \cdot \one^2_{F}(\zee)}}{\tildeEx{\one^2_{F}(\zee)}}
\]
where $p$ is any polynomial of degree at most $t-2k$.
\end{definition}

We can also define pseudocovariances and pseudo-variances for the conditioned SoS solutions.
\begin{definition}
	Let $F\sub [q]^m$ be an event such that $\one_F$ is $k$-local, and let $\tildeEx{\cdot}$ be a pseudoexpectation operator of degree $t$, with $t>2k$. Let $S,T$ be two sets of size at most $\frac{t-2k}{2}$ each. Then the pseudocovariance between $\zee_{S,\alpha}$ and $\zee_{T,\beta}$ for the solution conditioned on event $F$ is defined as,
\ifnum\confversion=1
\begin{multline*}
\tildecov(\zee_{S,\alpha},\zee_{T,\beta} \vert F) = \\
\tildeEx{ \zee_{S,\alpha} \zee_{T,\beta} \vert F} - \tildeEx{\zee_{S,\alpha} \vert F} \tildeEx{\zee_{T,\beta} \vert F}
\end{multline*}
\else
\begin{align*}
\tildecov(\zee_{S,\alpha},\zee_{T,\beta} \vert F) 
~=~ \tildeEx{ \zee_{S,\alpha} \zee_{T,\beta} \vert F} - \tildeEx{\zee_{S,\alpha} \vert F} \tildeEx{\zee_{T,\beta} \vert F}
\end{align*}
\fi
\end{definition}

We also define the pseudocovariance between $\zee_{S,\alpha}$ and $\zee_{T,\beta}$ after
conditioning on a random assignment for some $\zee_V$ with $V\sub [m]$. 
Note that here the random assignment for $\zee_V$ is chosen according to the local distribution for
the set $V$.

\begin{definition}[Pseudocovariance for conditioned pseudoexpectation operators]
\ifnum\confversion=1
\begin{multline*}
\tildecov(\zee_{S,\alpha},\zee_{T,\beta} \vert \zee_V) 
= \\ \sum_{\gamma \in [q]^V}   \tildecov(\zee_{S,\alpha},\zee_{T,\beta} \vert \zee_V = \gamma) \cdot \tildeEx{\zee_{V,\gamma}}
	\end{multline*}
\else
\begin{align*}
\tildecov(\zee_{S,\alpha},\zee_{T,\beta} \vert \zee_V) 
~=~ \sum_{\gamma \in [q]^V}   \tildecov(\zee_{S,\alpha},\zee_{T,\beta} \vert \zee_V = \gamma) \cdot \tildeEx{\zee_{V,\gamma}}
\end{align*}
\fi
\end{definition}

And we likewise define $\tildeVar{\zee_{S,\alpha} \vert \zee_V}$, $\tildecov(\zee_S, \zee_T \vert \zee_V)$ and $\tildeVar{\zee_S \vert \zee_V}$.

\subsection{SoS relaxations for codes}
For both Tanner codes and AEL codes, we will identify $[m]$ with $E$ so that the SoS solutions will be relaxations to the assignments to edges of a bipartite $(n,d,\lambda)$-expander.

\paragraph{Pseudocodewords for Tanner Codes}

Let $G(L,R,E)$ be the bipartite $(n,d,\lambda)$-expander on which the Tanner code is defined, and let $\calC_0\subseteq [q]^d$ be the inner code. The SoS variables will be $\zee = \{Z_{e,j}\}_{e\in E,j\in [q]}$. 

\begin{definition}[Tanner Pseudocodewords]
For $t\geq 2d$, we define a degree-$t$ Tanner pseuocodeword to be a degree-$t$ pseudoexpectation operator $\tildeEx{\cdot}$ on $\zee$ respecting the following constraints:
\ifnum\confversion=1
\begin{align*}
&\forall \li \in L, \quad \zee_{N_L(\li)} \in \calC_0 \\
\text{and } \quad &\forall \ri \in R,\quad \zee_{N_R(\ri)} \in\calC_0
\end{align*}
\else
\begin{align*}
\forall \li \in L, \quad \zee_{N_L(\li)} \in \calC_0 
\qquad \text{and} \qquad          
\forall \ri \in R,\quad \zee_{N_R(\ri)} \in\calC_0
\end{align*}
\fi
Again, since these constraints are $d$-local, it is sufficient to enforce that certain degree-$d$ polynomials are zero (respected by pseudoexpectation) to enforce these constraints. 
In particular, each parity check in the parity check matrix of $\calC_0$ will correspond to a monomial of size at most $d$ that we can enforce to be equal to 1.
\end{definition}

We can also define a generalization of distance between two codewords to include pseudocodewords.
\begin{definition}[Distance from a pseudocodeword]
The distance between a pseudocodeword $\tildeEx{\cdot}$ of SoS-degree $t\geq 2d$ and a codeword $h$ of $\TanC$ is defined as
\ifnum\confversion=1
\begin{align*}
\dis(\tildeEx{\cdot},h) 
~\defeq~& \Ex{e\in E}{\tildeEx{\indi{\zee_e \neq h_e}}} \\
~=~& \Ex{\li\in L}{\tildeEx{\dis(\zee_{N_L(\li)},  h_{N_L(\li)})}}
\end{align*}
\else
\[
\dis(\tildeEx{\cdot},h) 
~\defeq~ \Ex{e\in E}{\tildeEx{\indi{\zee_e \neq h_e}}} 
~=~ \Ex{\li\in L}{\tildeEx{\dis(\zee_{N_L(\li)},  h_{N_L(\li)})}}
\]
\fi
\end{definition}

\paragraph{Pseudocodewords for AEL}

Let $G(L,R,E)$ be the bipartite $(n,d,\lambda)$-expander on which the AEL code is defined, and let $\calC_0\subseteq [q_0]^d$ be the inner code. The SoS variables will be $\zee = \{Z_{e,j}\}_{e\in E,j\in [q_0]}$. 

\begin{definition}[AEL Pseudocodewords]
	For $t\geq 2d$, we define a degree-$t$ AEL pseuocodeword to be a degree-$t$ pseudoexpectation operator $\tildeEx{\cdot}$ on $\zee$ respecting the following constraints:
	\begin{align*}
		\forall \li \in L,\quad \zee_{N_L(\li)} \in \calC_0
	\end{align*}
\end{definition}

Next we define the distances between a pseudocodeword and a codeword of $\AELC$. 
\begin{definition}[Distance from a pseudocodeword]
	The left, middle and right distances between a pseudocodeword $\tildeEx{\cdot}$ of SoS-degree $t\geq 2d$ and a codeword $h \in \AELC$ are defined as
	\begin{align*}
		\dis^L(\tildeEx{\cdot},h) &\defeq \Ex{\li}{\tildeEx{\indi{\zee_{N_L(\li)} \neq h_{N_L(\li)}}}} \\
	\dis(\tildeEx{\cdot},h) &\defeq \Ex{e}{\tildeEx{\indi{\zee_e \neq h_e}}} \\
		\dis^R(\tildeEx{\cdot},h) &\defeq \Ex{\ri\in R}{\tildeEx{\indi{\zee_{N_R(\ri)}\neq  h_{N_R(\ri)}}}}
	\end{align*}
\end{definition}
%


\section{Proof Overview}
\label{sec:overview}

Our proof can be viewed as an algorithmic implementation of the proof of Johnson bound, and the
proofs of distance, for the relevant codes.
We start with an overview of the proof for Tanner codes. The argument is very similar for all codes
considered here, substituting an appropriate proof of distance in each case. 
\vspace{-5 pt}
\paragraph{Johnson bound via covering lemmas.}
We first prove the Johnson bound via a statement we will call a ``covering lemma'', which can then
be generalized to work with convex relaxations.
Given an $[m, \delta, \rho]_q$ code $\calC$ and a received word $g \in [q]^m$, our goal is to
bound the list of radius $\calL(g,\calJ_q(\delta)) = \inbraces{h\in \calC \suchthat \dis(h,g) <
  \calJ_q(\delta) }$, where $\calJ_q(\delta)$ denotes the Johnson bound. 
It will be convenient to work with $\beta = 1 -  \frac{q \cdot \delta}{q-1}$, and the Johnson bound can be
defined using the equation
\[
\inparen{1 - \frac{q \cdot \calJ_q(\delta)}{q-1}}^2 ~=~ \inparen{1 - \frac{q \cdot \delta}{q-1}} ~=~ \beta \mper
\]
Also, we can map elements of $[q]$ to corners $u_1, \ldots, u_q$ of the simplex in $\R^{q-1}$, which are unit
vectors satisfying $\ip{u_i}{u_j} = \nfrac{-1}{(q-1)}$ for $i \neq j$. Applying
this map, say $\chi$, pointwise to $g, h \in [q]^m$, gives
\ifnum\confversion=1
\begin{align*}
\ip{\chi(g)}{\chi(h)} ~=~& \Ex{i \in [m]}{\ip{\chi\inparen{g(i)}}{\chi\inparen{h(i)}}} \\
 ~=~& 1 - \Delta(g,h) - \frac{\Delta(g,h)}{q-1} \\
 ~=~& 1 - \frac{q \cdot \Delta(g,h)}{q-1} \mper
\end{align*}
\else
\[
\ip{\chi(g)}{\chi(h)} ~=~ \Ex{i \in [m]}{\ip{\chi\inparen{g(i)}}{\chi\inparen{h(i)}}} ~=~ 1 - \Delta(g,h) - \frac{\Delta(g,h)}{q-1} ~=~ 1 - \frac{q \cdot
  \Delta(g,h)}{q-1} \mper
\]
\fi
Thus, given $g$ and $\beta$ as above, we can write the list as 
$\calL = \inbraces{h\in \calC \suchthat \ip{\chi(g)}{\chi(h)} > \sqrt{\beta}}$. 
The covering lemma \ifnum\confversion=0 in \cref{sec:covering} \fi shows that given a set $\calF$ of unit vectors 
in an inner-product space, and a unit vector $\tilde{g}$ satisfying $\ip{\tilde{g}}{f} > \sqrt{\beta}$ for
all $f \in \calF$, there exists $g_0$ in the \emph{convex hull} $\conv{\calF}$ satisfying
$\ip{g_0}{f} > \beta$ for all $f \in \calF$. \ifnum\confversion=1 (See the full version for a proof based on minimizing $\ell_2$ norm of $g_0$ subject to $g_0 \in \conv{\calF}$.)\fi

Instantiating this with $\tilde{g} = \chi(g)$ and $\calF = \inbraces{\chi(h) \suchthat h \in
  \calL}$ gives $g_0 \in \conv{\calF}$ satisfying $\ip{g_0}{\chi(h)} > \beta$ for all $h
\in \calL$. 
Additionally, $g_0$ can be chosen to be supported on at most $m \cdot (q-1) +1$ elements of $\calF$ via
\Caratheodory theorem.
Since $\ip{\chi(h_1)}{\chi(h_2)} \leq \beta$ for any $h_1 \neq h_2$ in $\calC$, any $h \in \calF
\setminus \supp(g_0)$ will satisfy $\ip{g_0}{\chi(h)} \leq \beta$, which is a contradiction. 
Thus, we must have $\calL \subseteq \supp(g_0)$ implying $\abs{\calL} \leq (q-1) \cdot m + 1$.
While this is a weaker bound on the list size (but can be independent of $m$ via approximate \Caratheodory
theorems), each step of the above proof can be extended to work well with convex relaxations.
\vspace{-5 pt}
\paragraph{Algorithmic covering lemmas via SoS.}
Note that in the above proof, it suffices to have $g_0$ in the convex hull of \emph{all} codewords
\ie $g_0 \in \conv{\chi(h) \suchthat h \in \calC}$, instead of only the codewords in $\calL$.
Since the convex hull of all the codewords is still a difficult set to optimize over, we instead
work with degree-$t$ \emph{pseudocodewords} defined as solutions to an SoS relaxation of degree-$t$,
respecting certain local constraints corresponding to the code (see \cref{sec:prelims}).

Moreover, the covering lemma used above can be proved by finding the $g_0 \in \conv{\calF}$, which
minimizes $\norm{g_0}$ while satisfying $\ip{g_0}{g} > \sqrt{\beta}$.
Analogously, we consider solutions to the SoS relaxation, given as pseudocodewords $\tildeEx{\cdot}$
satisfying $\ip{\tildeEx{\chi(\zee)}}{\chi(g)} > \sqrt{\beta}$ and minimizing
$\norm{\tildeEx{\chi(\zee)}}$. Note that here $\chi(\zee)$ is a (vector-valued) 1-local function, and
thus the vector $\tildeEx{\chi(\zee)}$ is well defined.
A similar relaxation was also used for list-decoding of direct-sum codes in \cite{AJQST20}.
\vspace{-5 pt}
\paragraph{Choosing from the ``support'' via conditioning.}
The next step of the proof of Johnson bound can be seen as saying that for all $h_1$ such that
$\chi(h_1) \in \supp(g_0)$ and $h_2 \in \calL$, we must have $\ip{\chi(h_1)}{\chi(h_2)} = 1$ (when
they are equal) or $\ip{\chi(h_1)}{\chi(h_2)} \leq \beta$ (when they are distinct codewords).

We will do this in two steps. The first is to develop a good proxy for ``$\chi(h_1) \in \supp(g_0)$''
since we are working with the pseudocodeword $\tildeEx{\cdot}$, which is not a
convex combination of codewords.
Instead, it suffices to look at pseudocodewords which have small
pseudo-covariance across a typical pair of left-right vertices in the bipartite graph defining the Tanner code \ie
\[
\Ex{\substack{\ell \in L \\ r \in R}}{\tildecov(\zee_{N_L(\ell)},\zee_{N_R(r)})} ~\leq~ \eta \mper
\]
Note that when $\tildeEx{\cdot}$ corresponds to an actual codeword (or any
integral solution) the covariance will be 0. Thus, the above is a weakening of the notion of
``vertex of the convex hull''.  We call such a solution, an \emph{$\eta$-good} pseudocodeword.
A similar definition was also used by Richelson and Roy~\cite{RR22} in their list-decoding algorithm
for Ta-Shma's codes~\cite{TS17}, and also in earlier works~\cite{AJQST20, JQST20}.

An argument of Barak, Raghavendra, and Steurer~\cite{BRS11} shows that conditioning the starting SoS
solution on the values few randomly chosen variables, leads to an $\eta$-good solution. 
In fact, one either has small covariance in the sense above, or conditioning on the variables in
$N_R(r)$ for a random $r \in R$ reduces the average variance (seen from the left) by
$\Omega_{q,d}(\eta^2)$. Since the (pseudo-)variance is non-negative, this process terminates,
yielding an $\eta$-good solution.
The argument can also be made deterministic by enumerating over all (constant-sized) subsets to
condition on.

\vspace{-5 pt}
\paragraph{Proof of distance for good pseudocodewords.}
Given an $\eta$-good pseudocodeword $\tildeEx{\cdot}$ as defined above, we can now prove that for
any $h \in \calL$, we must have 
\ifnum\confversion=1
\begin{align*}
\ip{\tildeEx{\chi(\zee)}}{\chi(h)} &~\geq~ 1 - O(\eta) \\
\text{or} \qquad
\ip{\tildeEx{\chi(\zee)}}{\chi(h)} &~\leq~ \beta + O(\eta) \mper
\end{align*}
\else
\[
\ip{\tildeEx{\chi(\zee)}}{\chi(h)} ~\geq~ 1 - O(\eta) \qquad \text{or} \qquad
\ip{\tildeEx{\chi(\zee)}}{\chi(h)} ~\leq~ \beta + O(\eta) \mper
\]
\fi
To argue the above dichotomy, we will switch to the distances instead of inner products, as the
proof closely follows the distance proof of Tanner codes. Let $\inbraces{X_{\ell}}_{\ell \in L}$ and
$\inbraces{Y_r}_{r \in R}$ be ensembles of $d$-local functions defined as 
$X_{\ell}(\zee) = \one\inbraces{\zee_{N_L(\ell)} \neq h_{N_{L}(\ell)}}$ 
and 
$Y_{r}(\zee) = \one\inbraces{\zee_{N_R(r)} \neq h_{N_{R}(r)}}$. Using the fact that the distance of
the base code $\calC_0$ is at least $\delta_0$, and the fact that the local constraints for $\calC_0$
are respected by the pseudocodewords, one can show that 
\ifnum\confversion=1
\begin{align*}
\Delta(\tildeEx{\cdot}, h) 
~\defeq~& \Ex{e \in E}{\tildeEx{\one\inbraces{\zee_e \neq h_e}}} \\
~\geq~& \delta_0 \cdot \Ex{\ell \in L}{\tildeEx{X_{\ell}(\zee)}} \mper
\end{align*}
\else
\[
\Delta(\tildeEx{\cdot}, h) 
~\defeq~ \Ex{e \in E}{\tildeEx{\one\inbraces{\zee_e \neq h_e}}} 
~\geq~ \delta_0 \cdot \Ex{\ell \in L}{\tildeEx{X_{\ell}(\zee)}} \mper
\]
\fi
Taking the geometric mean with a similar inequality for $\inbraces{Y_r}_{r \in R}$ gives
\[
\Delta(\tildeEx{\cdot}, h) 
~\geq~ 
\delta_0 \cdot \inparen{\Ex{\substack{\ell \in L \\ r \in R}}{\tildeEx{X_{\ell}(\zee)} \cdot \tildeEx{Y_{r}(\zee)}}}^{1/2}
\mper
\]
On the other hand, a variant of the expander mixing lemma for pseudoexpectations (see \ifnum\confversion=1 full version\else \cref{sec:distance}\fi), together with the $\eta$-good property and the simple observation that \[
\indi{\zee_e \neq h_e} \leq X_{\li}(\zee) \cdot Y_{\ri}(\zee), \] gives
\ifnum\confversion=1
\small
\begin{align*}
&\ \Delta(\tildeEx{\cdot}, h) \\
&\leq \Ex{\ell \sim r}{\tildeEx{X_{\ell}(\zee) \cdot Y_r(\zee)}}  \\
&\leq \Ex{\ell \in L \atop r \in R}{\tildeEx{X_{\ell}(\zee) \cdot Y_r(\zee)}} + \\
 &\hspace*{2cm} \lambda \cdot \inparen{\Ex{\ell \in L \atop r \in R}{\tildeEx{X_{\ell}^2(\zee)} \cdot \tildeEx{Y_{r}^2(\zee)}}}^{1/2} \\
&\leq \Ex{\ell \in L \atop r \in R}{\tildeEx{X_{\ell}(\zee)} \cdot \tildeEx{Y_r(\zee)}} + \\
& \hspace*{2cm} \eta + \lambda \cdot \inparen{\Ex{\ell \in L \atop r \in R}{\tildeEx{X_{\ell}(\zee)} \cdot
       \tildeEx{Y_{r}(\zee)}}}^{1/2} \mper
\end{align*}
\normalsize
\else
\begin{align*}
\Delta(\tildeEx{\cdot}, h) 
&~\leq~ \Ex{\ell \sim r}{\tildeEx{X_{\ell}(\zee) \cdot Y_r(\zee)}}  \\
&~\leq~ 
 \Ex{\substack{\ell \in L \\ r \in R}}{\tildeEx{X_{\ell}(\zee) \cdot Y_r(\zee)}} 
~+~ 
\lambda \cdot \inparen{\Ex{\substack{\ell \in L \\ r \in R}}{\tildeEx{X_{\ell}^2(\zee)} \cdot
                                                                         \tildeEx{Y_{r}^2(\zee)}}}^{1/2}
  \\
&~\leq~
 \Ex{\substack{\ell \in L \\ r \in R}}{\tildeEx{X_{\ell}(\zee)} \cdot \tildeEx{Y_r(\zee)}} ~+~ \eta   
~+~ 
\lambda \cdot \inparen{\Ex{\substack{\ell \in L \\ r \in R}}{\tildeEx{X_{\ell}(\zee)} \cdot
       \tildeEx{Y_{r}(\zee)}}}^{1/2} \mper
\end{align*}
\fi
Note that the last line also used $\tildeEx{X_{\ell}^2(\zee)} = \tildeEx{X_{\ell}(\zee)}$ (and
similarly for $Y_r$).
Comparing the two bounds and solving the resulting quadratic inequality in $\tau =
(\Ex{\ell,r}{\tildeEx{X_{\ell}(\zee)} \cdot \tildeEx{Y_r(\zee)}})^{1/2}$ yields the dichotomy.
\vspace{-5 pt}
\paragraph{Completing the argument.}
Starting the algorithmic covering lemma with parameter $\beta + \eps$ ensures that with positive
(constant) probability, for any fixed $h \in \calL$, the conditioning yields an $\eta$-good
pseudocodeword satisfying $\ip{\tildeEx{\chi(\zee)}}{\chi(h)} \geq \beta + \nfrac{\eps}{2}$. 
By the above dichotomy, we must be in the first case (for sufficiently small $\eta$). 
A simple averaging argument then shows that considering $h' \in [q]^m$ with $h'(e) = \argmax_{j \in
  [q]}\inbraces{\tildeEx{\indi{\zee_{e}=j}}}$ gives $\Delta(h',h) = O(\eta)$.
Given such an $h'$, the codeword $h$ can be recovered by unique decoding.

The above argument can also be made to work for other codes, where the distance proof is a spectral argument
based on expanders, such as the codes obtained via the distance amplification of Alon, Edmonds and
Luby~\cite{AEL95}. 
One can simply use pseudocodewords for the appropriate code, and substitute the corresponding distance
proof in the argument above.


\section{Covering Lemma and Johnson Bounds}
\label{sec:covering}
In this section, we will introduce our abstract covering lemma, and then use it to give algorithm-friendly proofs of (known) Johnson bounds \cite{G01} by showing that there is a distribution over codewords that covers the list. Next, we will show that if we are willing to work with distributions over degree-$t$ pseudocodewords (or just pseudocodewords, due to convexity), then such a pseudocodeword may be found in time $n^{\calO(t)}$.

\subsection{Covering Lemma}

\begin{lemma}[Covering Lemma]\label{lem:covering}
	Let $\calH$ be an inner product space, and let $\calF$ be a family of unit vectors in $\calH$. Suppose there exists a $g \in\calH$ of unit norm such that for any $f\in \calF$, $\ip{g}{f} > \eps$. Then, there exists a $g_0 \in \conv{\calF}$ such that for any $f\in\calF$, $\ip{g_0}{f} > \eps^2$.
\end{lemma}

\begin{proof}
Consider the set $T = \inbraces{v\in \conv{\calF} \suchthat \ip{v}{g} >\eps}$, which is non-empty because any $f\in \calF$ also belongs to $T$. Let $g_0 = \argmin_{v\in T} \norm{v}^2$. We will show that $g_0$ must have the property that $\ip{g_0}{f}>\eps^2$ for any $f\in\calF$. Note that $\norm{g_0}\cdot \norm{g}\geq \ip{g_0}{g}>\eps$ means $\norm{g_0}>\eps$.

Suppose not, then there exists an $f\in\calF$ such that $\ip{g_0}{f} \leq \eps^2$. For an $\tee \in[0,1]$ to be chosen later, consider $g_1 = \tee g_0+(1-\tee)f$, which is also in $\conv{\calF}$, and $\ip{g_1}{g} = \tee \cdot \ip{g_0}{g}+(1-\tee)\cdot \ip{f}{g} >\eps$.
\begin{align*}
	\norm{g_1}^2 &~=~ \ip{\tee g_0+(1-\tee)f}{\tee g_0+(1-\tee)f} \\
	&~=~ \tee^2 \cdot \norm{g_0}^2+2\tee(1-\tee) \cdot \ip{g_0}{f}+(1-\tee)^2 \cdot \norm{f}^2\\
	&~\leq~ \tee^2 \cdot \norm{g_0}^2+2\tee(1-\tee)\eps^2+(1-\tee)^2\\
	&~=~\tee^2\cdot (\norm{g_0}^2-\eps^2) +\eps^2 \cdot (\tee+(1-\tee))^2 + (1-\tee)^2 (1-\eps^2)\\
	&~=~\tee^2\cdot (\norm{g_0}^2-\eps^2) +\eps^2 + (1-\tee)^2 \cdot (1-\eps^2)\\
	&~=~\tee^2 \cdot (1-\eps^2 + \norm{g_0}^2-\eps^2) - 2\tee \cdot (1-\eps^2) +1
\end{align*}
At $\tee=1$, this expression equals $\norm{g_0}^2$. The minimum of this quadratic function of $\tee$
is achieved at $\tee= \frac{2(1-\eps^2)}{2(1-\eps^2+\norm{g_0}^2-\eps^2)}$, which is $<1$ as
$\norm{g_0} >\eps$. Thus, we can reduce $\norm{g_1}^2$ further to strictly less than $\norm{g_0}^2$
by choosing $\tee$ to the above value, which contradicts the optimality of $g_0$.
\end{proof}
\subsection{Johnson bounds}

We will prove the standard $q$-ary Johnson bound and the version from \cite{GS01:johnson} which receives weight $W_{i,j}$ for each $j\in [q]$ for every coordinate $i\in [n]$. First we define some functions to embed the received word $g$ (or received weights) in $\R^{(q-1)n}$, which will be the inner product space where we apply the covering lemma.

\begin{definition}
Fix $q \in \N$. We denote by $\chi_q: [q]\rightarrow \R^{q-1}$ any function that satisfies
\begin{align*}
	\ip{\chi_q(j_1)}{\chi_q(j_2)} = \begin{cases}
		1 &j_1=j_2 \\
		\frac{-1}{q-1} &j_1\neq j_2
	\end{cases}
\end{align*} 
For $f\in [q]^n$, we use $\chi_q(f)$ to denote the vector in $\R^{(q-1)n}$ obtained by applying $\chi_q$ on $f$ coordinate-wise. Note that \[\norm{\chi_q(f)}^2 ~=~ \Ex{i\in[n]}{\norm{\chi_q(f_i)}^2} = 1\]\\
When clear from context, we may omit $q$ to write $\chi(f)$.
\end{definition}

Such a $\chi_q$ exists because the corresponding $q\times q$ Gram matrix is positive semidefinite and of rank $q-1$.

\begin{observation}
For $f_1,f_2 \in [q]^n$, if $\dis(f_1,f_2) = (1-\frac{1}{q})(1-\beta)$, then
\[
	\ip{\chi(f)}{\chi(g)} = \beta
\]
\end{observation}

\begin{theorem}[$q$-ary Johnson bound]\label{thm:johnson_bound}
Let $\calC \subseteq [q]^n$  be a code of distance at least $\dis(\calC) = (1-\frac{1}{q})(1-\beta)$. For any $g\in [q]^n$, 
\[
	\abs{\calL \inparen{ g,\inparen{1-\frac{1}{q}} \cdot (1-\sqrt{\beta}) } } ~\leq~ (q-1)\cdot n+1
\]
\end{theorem}

\begin{proof}

Let $\calL = \calL \inparen{ g,\inparen{1-\frac{1}{q}}\cdot (1-\sqrt{\beta}) } = \inbraces{h_1,h_2,\cdots,h_m}$. Then for any $h_i\in \calL$,
\[
	\ip{\chi(g)}{\chi(h_i)} > \sqrt{\beta}
\]
By \cref{lem:covering}, there exists a $g_0 \in \conv{\inbraces{\chi(h_1),\chi(h_2),\cdots,\chi(h_m)}}$ such that $\ip{g_0}{\chi(h)}>\beta$ for every $h\in \calL$. Since $g_0 \in \R^{(q-1)n}$, by the \Caratheodory theorem, we may assume that $g_0$ can be written as a convex combination of at most $(q-1)n+1$ elements of $\inbraces{\chi(h_1),\chi(h_2),\cdots,\chi(h_m)}$. Let this subset of $\inbraces{\chi(h_1),\chi(h_2),\cdots,\chi(h_m)}$ be $\supp(g_0)$. 

For any $h\in \calL$,
\begin{align*}
	\ip{g_0}{\chi(h)} > \beta &~\implies~
	\exists \chi(h_0)\in \supp(g_0) \text{ such that } \ip{ \chi(h_0) }{\chi(h)} > \beta \\
	&~\implies~ \exists \chi(h_0)\in \supp(g_0) \text{ such that } \ip{\chi(h_0)}{\chi(h)}=1 \\
	&~\implies~ \exists \chi(h_0)\in \supp(g_0) \text{ such that } h_0=h \\
	&~\implies~ \chi(h) \in \supp(g_0)
\end{align*}
Here we used the fact that for any $h_i,h_j\in \calL$, due to distance of the code,
\[
	\ip{\chi(h_i)}{\chi(h_j)} \leq \beta \quad \text{ or }\quad \ip{\chi(h_i)}{\chi(h_j)} =1
\]
Finally, we have concluded that $\inbraces{\chi(h_1),\chi(h_2),\cdots,\chi(h_m)} \sub \supp(g_0)$ and using $|\supp(g_0)| \leq (q-1)\cdot n+1$ gives $m=|\calL| \leq (q-1)\cdot n +1$.
\end{proof}

Note that in the above proof, $g_0\in \conv{\inbraces{\chi(h_1),\chi(h_2),\cdots,\chi(h_m)}} \sub \conv{\chi(\calC)}$. Indeed, even if we minimize the norm as in the proof of \cref{lem:covering} over $\conv{\chi(\calC)}$, we will still get the covering property. This will be useful when trying to find the cover via an efficient algorithm, where we do not know the list apriori.

We now prove the more general weighted version of Johnson bound, which captures list recovery as a special case.

\begin{theorem}[Weighted Johnson bound \cite{G01}]\label{thm:weighted_johnson_bound}
	Let distance of code be at least $(1-\frac{1}{q})(1-\beta)$. Given weights $\{w_{i,j}\}_{i\in [n],j\in [q]}$, let $W_i = \sum_j w_{i,j}$ and $W_i^{(2)} = \sum_j w_{i,j}^2$. The number of codewords $h$ that satisfy
	\[
		\Ex{i}{\frac{w_{i,h_i}}{W_i}} > \frac{1}{q} + \sqrt{\inparen{1-\frac{1}{q}} \cdot
                  \Ex{i}{\frac{W_i^{(2)}}{W_i^2} -\frac{1}{q}} \cdot \beta }
	\]
	is at most $n(q-1)+1$.
\end{theorem}
\begin{proof}
Embed the received weights in $i^{th}$ position as $\chi(\{w_{i,j}\}_{j\in [q]}) = \sum_j \frac{w_{i,j} \chi_q(j)}{W_i}$, and append all these $n$ vectors, each of dimension $q-1$, and normalize to a unit vector to form the final embedded vector $u$. This normalizing factor is $ \inparen{\ExpOp_{i}{\norm{\frac{\sum_j w_{ij} \chi_q(j)}{W_i} }^2}}^{1/2} $, which we can simplify as
\begin{align*}
	\Ex{i}{\norm{\frac{\sum_j w_{ij} \chi_q(j)}{W_i} }^2} &~=~ \Ex{i}{\frac{1}{W_i^2} \left( \sum_j w_{ij}^2 -\frac{1}{q-1} \sum_{j_1\neq j_2} w_{ij_1}w_{ij_2}\right) }	\\
	&~=~ \Ex{i}{\frac{1}{W_i^2} \left( \frac{q}{q-1} \sum_j w_{ij}^2 -\frac{1}{q-1} (\sum_{j} w_{ij})^2 \right) }\\
	&~=~ \frac{q}{q-1} \cdot \Ex{i}{ \frac{1}{W_i^2} \left( W_i^{(2)} -\frac{1}{q}W_i^2 \right) }\\
	&~=~ \frac{q}{q-1} \cdot \Ex{i}{ \frac{W_i^{(2)}}{W_i^2} -\frac{1}{q} }
\end{align*}
We then have
\begin{align*}
	 \inparen{\frac{q}{q-1} \cdot \Ex{i}{ \frac{W_i^{(2)}}{W_i^2} -\frac{1}{q} }}^{1/2} \cdot  ~\ip{u}{\chi_q(h)} &~=~ \Ex{i}{\frac{1}{W_i} \left( \sum_{j} w_{i,j} \ip{\chi_q(j)}{\chi_q(h_i)}\right) } \\
	&~=~ \Ex{i}{\frac{1}{W_i} \left(  w_{i,h_i} - \frac{\sum_{j\neq f_i} w_{i,j}}{q-1} \right)} \\
	&~=~ \Ex{i}{ \frac{q}{q-1} \frac{w_{i,h_i}}{W_i} - \frac{1}{q-1} } 
	~=~ \frac{q}{q-1} \cdot \Ex{i}{ \frac{w_{i,h_i}}{W_i} - \frac{1}{q} } 
\end{align*}

For a codeword $h$ satisfying $\Ex{i}{\frac{w_{i,h_i}}{W_i}} > \frac{1}{q} +
\sqrt{(1-\frac{1}{q}) \cdot \Ex{i}{\frac{W_i^{(2)}}{W_i^2} -\frac{1}{q}} \cdot \beta }$, we
have $\ip{u}{\chi(h)} >\sqrt{\beta}$
and the rest of the proof follows as in \cref{thm:johnson_bound}, by using the fact that any two codewords $h_1,h_2$ satisfy $\ip{\chi(h_1)}{\chi(h_2)} \leq \beta$.
\end{proof}

\subsection{Algorithmic covering lemma}

As mentioned earlier, the above lemma guarantees the existence of a distribution over the list that agrees simultaneously with the entire list, but finding such a distribution may be hard without knowing the list already! The other option is to look for a distribution over all codewords, but such a polytope will have exponential number of vertices and it's not clear how to optimize over it efficiently.

We instead relax to allow degree-$t$ pseudodistributions, represented as degree-$t$ pseudoexpectation operators, which can be optimized over in time $n^{O(t)}$.

We recall that $\chiTan: [q] \rightarrow \R^{q-1}$ is a function that is extended pointwise to $\chiTan:[q]^E \rightarrow \inparen{\R^{q-1}}^{|E|}$. Note that $\chiTan$ is a 1-local vector-valued function on $[q]^E$, and so we can use $\tildeEx{\chiTan(\zee)} \in \R^E$ which will satisfy
\[
	\tildeEx{\chiTan(\zee)}(e) = \tildeEx{\chi_q(\zee_e)}
\]

\subsubsection{Tanner code}

Let $q \geq 2$ be an integer. Fix $G(L,R,E)$ to be an $(n,d,\lambda)$-expander, and $\calC_0 \subseteq [q]^d$ to be an inner code. Let $\TanC \subseteq [q]^E$ be the Tanner code determined by $G$ and $\calC_0$. 

\begin{lemma}\label{lem:abstract_algorithmic_covering}
	Fix $\gamma>0$. Let $u\in \R^{(q-1)|E|}$ with $\norm{u}_2 =1$. For any $t\geq d$, there exists a pseudocodeword $\tildeEx{\cdot}$ of SoS-degree $t$ such that for any $h \in \TanC$ that satisfies $\ip{u}{\chiTan(h)}>\gamma$, 
	\[
		\ip{\tildeEx{\chiTan(\zee)}}{\chiTan(h)} = \Ex{e\in E}{\ip{\tildeEx{\chi_q(\zee_e)}}{\chi_q(h(e))}} > \gamma^2
	\]
	Moreover, this pseudocodeword $\tildeEx{\cdot}$ can be found in time $n^{\bigoh(t)}$.
\end{lemma}
\begin{proof}
	Define the convex function $\Psi(\tildeEx{\cdot}) = \ip{\tildeEx{\chiTan(\zee)}}{\tildeEx{\chiTan(\zee)}}$ and solve the convex program in \cref{tab:Tanner_SoS}.
	
\begin{table}[h]
\hrule
\vline
\begin{minipage}[t]{0.99\linewidth}
\vspace{-5 pt}
{\small
\begin{align*}
    &\mbox{minimize}\quad ~~ \Psi\left(\tildeEx{\cdot}\right)
    \\
&\mbox{subject to}\quad \quad ~\\
    &\qquad \ip{\tildeEx{\chiTan(\zee)}}{u} ~>~ \gamma\label{cons:agreement-ld}    \\
&\qquad \tildeEx{\cdot} \text{is a pseudocodeword of SoS-degree } t
\end{align*}}
\vspace{-14 pt}
\end{minipage}
\hfill\vline
\hrule
\caption{Finding cover for the list $\calL$.}
\label{tab:Tanner_SoS}
\end{table}
	
	Let $\tildeEx{\cdot}$ be the pseudocodeword obtained, with $\Psi^*=\Psi(\tildeEx{\cdot})$. We will use the optimality of $\tildeEx{\cdot}$ to argue that for any $h\in \calL$,
	\[
		\ip{\tildeEx{\chiTan(\zee)}}{\chi(h)} > \gamma^2
	\]
	
	Suppose not. Then there exists $h\in \TanC$ such that $\ip{\tildeEx{\chiTan(\zee)}}{\chiTan(h)} \leq \gamma$. Then, for some $\tee \in[0,1]$ to be chosen later, consider 
	\[
		\PExp_{\tee}[\cdot] = (1-\tee)\tildeEx{\cdot} + \tee\PExp^{(h)}[\cdot]
	\]
	Here, we think of $h$ as an SoS-degree-$t$ pseudocodeword, so that by convexity, $\PExp_{\tee}[\cdot]$ is also an SoS-degree-$t$ pseudocodeword. We will show towards a contradiction that $\Psi(\PExp_{\tee}[\cdot]) < \Psi^*$.
	\begin{align*}
		\Psi(\PExp_{\tee}[\cdot])  &= \ip{\PExp_{\tee}[\chiTan(\zee)]}{\PExp_{\tee}[\chiTan(\zee)]} \\
		&= \ip{(1-\tee)\cdot \tildeEx{\chiTan(\zee)} + \tee\cdot \PExp^{(h)}[\chiTan(\zee)]}{(1-\tee)\cdot \tildeEx{\chiTan(\zee)}+\tee\cdot \PExp^{(h)}[\chiTan(\zee)]} \\
		&= \ip{(1-\tee) \cdot \tildeEx{\chiTan(\zee)} + \tee\cdot \chiTan(h)}{(1-\tee)\cdot \tildeEx{\chiTan(\zee)}+\tee \cdot \chiTan(h)} \\
		&= (1-\tee)^2 \cdot \ip{\tildeEx{\chiTan(\zee)}}{\tildeEx{\chiTan(\zee)}} +2\cdot \tee(1-\tee)\cdot \ip{\tildeEx{\chiTan(\zee)}}{\chiTan(h)} + \tee^2\cdot \ip{\chiTan(h)}{\chiTan(h)}\\
		&\leq (1-\tee)^2\cdot \Psi^* +2\cdot \tee(1-\tee)\cdot \gamma +\tee^2
	\end{align*}
	
	Optimizing over $\tee$, we choose $\tee^* = \frac{\Psi^* - \gamma}{\Psi^*-2\gamma+1}$. As long as $\Psi^* > \gamma$, we get optimal $\tee^* >0$, which implies $\Psi(\PExp_{\tee^*}[\cdot]) < \Psi\inparen{\PExp_{\tee}[\cdot]{\Big \vert}_{\tee=0}} = \Psi^*$, which would be a contradiction.
\end{proof}

\begin{lemma}\label{lem:cover_for_list_tanner}
	Fix $\eps>0$. Let the distance of $\TanC$ be $\delta$, and let $g$ be a received word.
	For any $t\geq d$, there exists a degree-$t$ pseudocodeword $\tildeEx{\cdot}$ such that for any $h\in \TanC$ such that $\dis(g,h)<\calJ(\delta)-\eps$,
	\[
		\dis(\tildeEx{\cdot},h) < \delta-2\eps \cdot \sqrt{1-\frac{q}{q-1} \cdot\delta}
	\]
	Moreover, this pseudocodeword $\tildeEx{\cdot}$ can be found in time $n^{\bigoh(t)}$.
\end{lemma}

\begin{proof}
	We use $g$ to construct a vector $u$, which we can use to find the required pseudocodeword via \cref{lem:abstract_algorithmic_covering}.
	
	Let $\delta = \inparen{1-\frac{1}{q}}(1-\beta)$ so that $\calJ(\delta) = \inparen{1-\frac{1}{q}}(1-\sqrt{\beta})$. Let $u = \chi(g)$. 
	
	For any $h\in \calL(g,\calJ(\delta)-\eps)$,
	\[
		\dis(g,h) < \calJ(\delta)-\eps = \inparen{1-\frac{1}{q}}(1-\sqrt{\beta}) -\eps \implies  \ip{\chi(g)}{\chi(h)} > \sqrt{\beta} + \frac{q}{q-1}\cdot \eps
	\]
	 Therefore, using \cref{lem:abstract_algorithmic_covering} and vector $u$, we can find a degree-$t$ pseudocodeword $\tildeEx{\cdot}$ such that for any $h\in \calL(g,\calJ(\delta)-\eps)$,
	 \[
	 	\ip{\tildeEx{\chiTan(\zee)}}{\chiTan(h)} > \inparen{\sqrt{\beta}+\frac{q}{q-1}\cdot \eps}^2 > \beta + 2\sqrt{\beta} \cdot \frac{q}{q-1} \cdot \eps
	 \]
	 Writing again in terms of distances and using $\beta = 1-\frac{q}{q-1}\delta$, this means,
	 \[
	 	\dis(\tildeEx{\cdot},h) < \inparen{1-\frac{1}{q}}(1-\beta -2\sqrt{\beta}\cdot
                \frac{q}{q-1}\cdot \eps) = \delta -2\sqrt{1-\frac{q}{q-1}\delta}\cdot \eps \mper
                \qquad\qquad \qquad \qedhere
	 \]
\end{proof}
	Next we show that \cref{lem:abstract_algorithmic_covering} can also be used to efficiently find a cover for the list when dealing with list recovery, by using the given weights to construct a modified vector $u$ (which will be the same vector embedding that was used in proof of \cref{thm:weighted_johnson_bound}).
\begin{lemma}\label{lem:cover_list_recovery_tanner}
	Fix $\eps>0$. Let the distance of $\TanC$ be $\delta$, and let the given weights be $\{w_{e,j}\}_{e\in E,j\in [q]}$. Assume that the weights are normalized so that $\sum_j w_{e,j}=1$ and denote  $W_e^{(2)} = \sum_j w_{e,j}^2$. \\
	For any $t\geq d$, there exists a pseudocodeword $\tildeEx{\cdot}$ of SoS-degree $t$ such that for any $h\in \TanC$ that satisfies
	\[
		\Ex{e}{w_{e,h(e)}} > \frac{1}{q} + \sqrt{\inparen{1-\frac{1}{q}  -\delta} \cdot \inparen{\Ex{e}{W_e^{(2)} -\frac{1}{q}}} } + \eps
	\]
	also satisfies $\dis(\tildeEx{\cdot},h) < \delta - \bigomega_{q,\delta,W}(\eps)$.
	
	Moreover, this pseudocodeword $\tildeEx{\cdot}$ can be found in time $n^{\bigoh(t)}$.
\end{lemma}

\begin{proof}
Again, let $\delta = \inparen{1-\frac{1}{q}}(1-\beta)$. We use the same vector $u$ to embed the given weights, with the property that for any $h$ that satisfies 
	\begin{align*}
		\Ex{e}{w_{e,h(e)}} &> \frac{1}{q} + \sqrt{\inparen{1-\frac{1}{q}  -\delta} \cdot \inparen{\Ex{e}{W_e^{(2)} -\frac{1}{q}}} } + \eps \\
		&= \frac{1}{q} + \sqrt{\inparen{1-\frac{1}{q}} \cdot \inparen{\Ex{e}{W_e^{(2)} -\frac{1}{q}}} \beta} + \eps
	\end{align*}
	now satisfies \[ \ip{u}{\chi(h)}> \sqrt{\beta} + \sqrt{\frac{q}{q-1}} \frac{ \eps }{ \sqrt{\Ex{e}{W_e^{(2)} -\frac{1}{q}}}} \]
Next, we appeal to \cref{lem:abstract_algorithmic_covering} with $u$ to efficiently obtain a degree-$t$ pseudocodeword $\tildeEx{\cdot}$ such that
	\[
		\ip{\tildeEx{\chiTan(\zee)}}{\chiTan(h)} > \beta + 2 \sqrt{\beta} \sqrt{\frac{q}{q-1}} \frac{ \eps }{ \sqrt{\Ex{e}{W_e^{(2)} -\frac{1}{q}}}}
	\]
	In terms of distance, this means,
\begin{align*}
\dis(\tildeEx{\cdot},h) &~<~ \inparen{1-\frac{1}{q}} \inparen{1-\beta -2\sqrt{\beta}
                          \sqrt{\frac{q}{q-1}} \frac{\eps }{ \sqrt{\Ex{e}{W_e^{(2)} -\frac{1}{q}}}} } \\
		&~=~ \delta -2 \eps \cdot \sqrt{1-\frac{q}{q-1}\delta}
           \cdot\sqrt{\frac{1-\frac{1}{q}}{\Ex{e}{W_e^{(2)} -\frac{1}{q}}}} \mper
                \qquad\qquad \qquad \qquad  \qedhere
	\end{align*} 
\end{proof}

\subsubsection{AEL Code}

Fix $q\in \N$. Fix $G=(L,R,E)$ to be an $(n,d,\lambda)$-expander, $\calC_0 \subseteq [q]^d$ to be an inner code, and $\calC_1 \subseteq [|\calC_0|]^n$ to be an outer code. Let $\AELC \subseteq [q^d]^R$ be the Tanner code determined by $G$, $\calC_0$ and $\calC_1$. 

\begin{definition}
	Let $\chiAEL: [q]^E \rightarrow \inparen{\R^{q^d-1}}^R$ be defined as
	\[
		\inparen{\chiAEL(f)}(\ri) = \chi_{q^d}(f_{N_R(\ri)})
	\]
	so that $\chiAEL$ is a $d$-local vector-valued function on $[q]^E$.
\end{definition}

\begin{observation}
If $f_1,f_2\in [q]^E$ with $\dis_{AEL}(f_1,f_2) = (1-\frac{1}{q^d})(1-\beta)$, then $\ip{\chiAEL( f_1)}{\chiAEL(f_2)} = \beta$.
\end{observation}

\begin{lemma}\label{lem:abstract_algorithmic_covering_ael}
	Fix $\gamma>0$. Let $u\in (\R^{q^d-1})^R$ with $\norm{u}=1$. For any $t\geq d$, there exists a degree-$t$ pseudocodeword $\tildeEx{\cdot}$ such that for any $h\in \AELC$ that satisfies $\ip{u}{\chiAEL(h)} > \gamma$, 
	\[
		\ip{\tildeEx{\chiAEL(\zee)}}{\chiAEL(h)} = \Ex{\ri \in R}{\ip{\tildeEx{\chi_{q^d}(\zee_{N_R(\ri)})}}{\chi_{q^d}(h(\ri))}} > \gamma^2
	\]
	Moreover, this pseudocodeword $\tildeEx{\cdot}$ can be found in time $n^{\bigoh(t)}$.
\end{lemma}

\begin{proof}
The proof is very similar to the proof of \cref{lem:abstract_algorithmic_covering}, except we replace embedding function $\chiTan$ by $\chiAEL$.

Define the quantity $\Psi(\tildeEx{\cdot}) = \ip{\tildeEx{\chiAEL(\zee)}}{\tildeEx{\chiAEL(\zee)}}$ and solve the following convex program.
	
\begin{table}[h]
\hrule
\vline
\begin{minipage}[t]{0.99\linewidth}
\vspace{-5 pt}
{\small
\begin{align*}
    &\mbox{minimize}\quad ~~ \Psi\left(\tildeEx{\cdot}\right)
    \\
&\mbox{subject to}\quad \quad ~\\
    &\qquad \ip{\tildeEx{\chiAEL(\zee)}}{\chiAEL(g)} >  \gamma\label{cons:agreement-ld}
    \\
&\qquad \tildeEx{\cdot} \text{is a pseudocodeword of SoS-degree } t
\end{align*}}
\vspace{-14 pt}
\end{minipage}
\hfill\vline
\hrule
\caption{Finding cover for the list $\calL$.}
\end{table}

Let $\tildeEx{\cdot}$ be the pseudocodeword obtained, with $\Psi^*=\Psi(\tildeEx{\cdot})$. Suppose it does not have the covering property. Then there exists $h\in \calL$ such that $\ip{\tildeEx{\chiAEL(\zee)}}{\chiAEL(h)} \leq \gamma$. Then, for some $\tee \in[0,1]$ to be chosen later, consider 
	\[
		\PExp_{\tee}[\cdot] = (1-\tee)\tildeEx{\cdot} + \tee\PExp^{(h)}[\cdot]
	\]
We again have,
\begin{align*}
		\Psi(\PExp_{\tee}[\cdot])  &= \ip{\PExp_{\tee}[\chiAEL(\zee)]}{\PExp_{\tee}[\chiAEL(\zee)]} \\
		&= \ip{(1-\tee) \tildeEx{\chiAEL(\zee)} + \tee\PExp^{(h)}[\chiAEL(\zee)]}{(1-\tee)\tildeEx{\chiAEL(\zee)}+\tee \PExp^{(h)}[\chiAEL(\zee)]} \\
		&= \ip{(1-\tee) \tildeEx{\chiAEL(\zee)} + x\chiAEL(h)}{(1-\tee)\tildeEx{\chiAEL(\zee)}+\tee\chiAEL(h)} \\
		&= (1-\tee)^2 \ip{\tildeEx{\chiAEL(\zee)}}{\tildeEx{\chiAEL(\zee)}} +2\tee(1-\tee)\ip{\tildeEx{\chiAEL(\zee)}}{\chiAEL(h)} + \tee^2\ip{\chiAEL(h)}{\chiAEL(h)}\\
		&\leq (1-\tee)^2 \Psi^* +2\tee(1-\tee)\gamma +\tee^2
	\end{align*}
	
	Optimizing over $\tee$, we choose $\tee^* = \frac{\Psi^* - \gamma}{\Psi^*-2\gamma+1}$. We will obtain a contradiction since $\Psi^* > \gamma^2$.
\end{proof}

\begin{lemma}\label{lem:cover_for_list_ael}
Fix $\eps>0$. Let the distance of $\AELC$ be $\delta$, and let $g$ be a received word. For any $t\geq d$, there exists a degree-$t$ pseudocodeword $\tildeEx{\cdot}$ such that for any $h \in \AELC$ such that $\dis(g,h) < \calJ(\delta) - \eps$,
\[
	\Ex{\ri}{\indi{ \tildeEx{\zee_{N_R(\ri)} \neq h_{\ri}}}} < \delta - 2 \eps \cdot \sqrt{1-\frac{q^d}{q^d-1} \cdot \delta}
\]
\end{lemma}

\begin{proof}
Same proof as the proof of \cref{lem:cover_for_list_tanner}, with the alphabet changed. The received word $g$ can be used to construct a unit vector $u = \chiAEL(g)$, which is then used via \cref{lem:abstract_algorithmic_covering_ael} to find a pseudocodeword with the required covering property.
\end{proof}

\begin{lemma}\label{lem:cover_for_list_recovery_ael}
	Fix $\eps>0$. Let the distance of $\AELC$ be $\delta$, and let the given weights be $\{w_{\ri,j}\}_{\ri\in R,j\in [q^d]}$. Assume that the weights are normalized so that $\sum_j w_{\ri,j}=1$ and denote  $W_{\ri}^{(2)} = \sum_j w_{{\ri},j}^2$. \\
	For any $t\geq d$, there exists a degree-$t$ pseudocodeword $\tildeEx{\cdot}$ such that for any $h\in \AELC$ that satisfies
	\[
		\Ex{\ri}{w_{\ri,h(\ri)}} > \frac{1}{q^d} + \sqrt{\inparen{1-\frac{1}{q^d}  -\delta} \cdot \inparen{\Ex{\ri}{W_{\ri}^{(2)} -\frac{1}{q^d}}} } + \eps
	\]
	also satisfies $\dis(\tildeEx{\cdot},h) < \delta - \bigomega_{q,d,\delta,W}(\eps)$.
	
	Moreover, this pseudocodeword $\tildeEx{\cdot}$ can be found in time $n^{\bigoh(t)}$.
\end{lemma}

\begin{proof}
	Same proof as the proof of \cref{lem:cover_list_recovery_tanner}.
\end{proof}


\section{Sum-of-Squares Proofs of Distance}
\label{sec:distance}
We will be proving that pseudocodewords satisfying certain $\eta$-good property defined below have the same distance properties as true codewords, up to $\eta$ error.
\begin{definition}
A pseudocodeword of SoS-degree at least $2d$ is $\eta$-good if
\begin{align*}
    \Ex{\li,\ri}{\tildecov[\zee_{N_L(\li)},\zee_{N_R(\ri)}]} \leq \eta
\end{align*}
\end{definition}

\begin{observation}
    A true codeword is a 0-good pseudocodeword.
\end{observation}

The $\eta$-good property is useful to change the pseudoexpectation of a product of functions into the product of pseudoexpectations of those functions. We establish a formal claim about this in the next lemma. The terms $\norm{X_{\li}}_{\infty}$ and $\norm{Y_{\ri}}_{\infty}$ should be just seen as normalizing the scale, and indeed we will only use functions that are bounded in infinity norm by 1.

\begin{lemma}
	Let $\{X_{\li}\}_{\li\in L}$ and $\{Y_{\ri}\}_{\ri\in R}$ be two collections of $d$-local functions on $[q]^E$ such that for every $\li\in L$, $X_{\li}(f)$ only depends on $f|_{N_L(\li)}$ and for every $\ri\in R$, $Y_{\ri}(f)$ only depends on $f|_{N_R(\ri)}$. Then, for an $\eta$-good pseudocodeword $\tildeEx{\cdot}$,
	\[
		\Ex{\li,\ri}{\tildeEx{X_{\li}(\zee)Y_{\ri}(\zee)}} ~\leq~ \Ex{\li,\ri}{\tildeEx{X_{\li}(\zee)}\tildeEx{Y_{\ri}(\zee)}} + \eta \left( \max_{\li} \norm{X_{\li}}_{\infty} \right) \left( \max_{\ri} \norm{Y_{\ri}}_{\infty}\right)
	\]
\end{lemma}

\begin{proof}
	For any $\li$ and $\ri$,
{\small
	\begin{align*}
		&\tildeEx{X_{\li}(\zee)Y_{\ri}(\zee)} - \tildeEx{X_{\li}(\zee)}\tildeEx{Y_{\ri}(\zee)} \\
		&~=~ \sum_{\substack{\alpha \in [q]^{N_L(\li)} \\ \beta\in[q]^{N_R(\ri)}}}
                     \tildeEx{X_{\li}(\alpha) \prod_{s\in N_L(\li)} \zee_{s,\alpha_s} \cdot Y_{\ri}(\beta) \prod_{t\in N_R(\ri)} \zee_{t,\beta_t}} - 
		 \sum_{\substack{\alpha \in [q]^{N_L(\li)} \\ \beta\in[q]^{N_R(\ri)}}} \cdot \tildeEx{X_{\li}(\alpha) \prod_{s\in N_L(\li)} \zee_{s,\alpha_s}} \tildeEx{Y_{\ri}(\beta) \prod_{t\in N_R(\ri)} \zee_{t,\beta_t}} \\
		&~=~ \sum_{\alpha,\beta} X_{\li}(\alpha) Y_{\ri}(\beta) \cdot \inparen{\tildeEx{
           \prod_{s\in N_L(\li)} \zee_{s,\alpha_s} \cdot \prod_{t\in N_R(\ri)} \zee_{t,\beta_t}} -
           \tildeEx{\prod_{s\in N_L(\li)} \zee_{s,\alpha_s}} \cdot \tildeEx{\prod_{t\in N_R(\ri)} \zee_{t,\beta_t}}} \\
		&~\leq~  \norm{X_{\li}}_{\infty} \norm{Y_{\ri}}_{\infty} \cdot \mathlarger{\sum}_{\alpha,\beta}
           \abs{ \tildeEx{ \prod_{s\in N_L(\li)} \zee_{s,\alpha_s} \prod_{t\in N_R(\ri)}
           \zee_{t,\beta_t}} - \tildeEx{\prod_{s\in N_L(\li)} \zee_{s,\alpha_s}} \cdot \tildeEx{\prod_{t\in N_R(\ri)} \zee_{t,\beta_t}}} \\
		&~=~  \norm{X_{\li}}_{\infty} \norm{Y_{\ri}}_{\infty} \cdot \tildecov[\zee_{N_L(\li)},\zee_{N_R(\ri)}]
	\end{align*}		
}	
	Averaging over $\li$ and $\ri$, we get
\begin{align*}
\Ex{\li,\ri}{\tildeEx{X_{\li}(\zee)Y_{\ri}(\zee)} - \tildeEx{X_{\li}(\zee)}\tildeEx{Y_{\ri}(\zee)}} 
	&~\leq~  \Ex{\li,\ri}{\norm{X_{\li}}_{\infty} \norm{Y_{\ri}}_{\infty} \cdot \tildecov[\zee_{N_L(\li)},\zee_{N_R(\ri)}]} \\
	&~\leq~ \max_{\li}\norm{X_{\li}}_{\infty} \cdot \max_{\ri} \norm{Y_{\ri}}_{\infty} \cdot \Ex{\li,\ri}{\tildecov[\zee_{N_L(\li)},\zee_{N_R(\ri)}]} \\
	&~\leq~ \eta \cdot \max_{\li}\norm{X_{\li}}_{\infty} \cdot \max_{\ri}
   \norm{Y_{\ri}}_{\infty}
\qquad \qquad \qquad \qquad  \qedhere
\end{align*}
\end{proof}
The proofs of distance for both Tanner and AEL codes go via Expander Mixing Lemma (EML), and so we establish the analog of EML for pseudocodewords. Morally speaking, EML allows us to change measure from the edges of an expander to the complete (bipartite) graph for product functions.  First we prove a version of EML for vector valued functions, and then show that since pseudoexpectation operators can be written in terms of certain underlying vectors, they also satisfy a version of EML. Note that this step does not require any $\eta$-good property.
	
	\begin{lemma}[EML for vector-valued functions]
	\label{high_dimensional_eml}
		Let $\{v_{\li}\}_{\li \in L}$ and $\{u_{\ri}\}_{\ri \in R}$ be a collection of vectors in $\R^N$. Then,
		\[
			\abs{\Ex{\li \sim \ri}{\ip{v_{\li}}{u_{\ri}}} - \Ex{\li,\ri}{\ip{v_{\li}}{u_{\ri}}}} \leq \lambda \sqrt{\Ex{\li}{\norm{v_{\li}}^2}} \sqrt{\Ex{\ri}{\norm{u_{\ri}}^2}}
		\]
	\end{lemma}
	
	\begin{proof}
		Usual EML applied coordinate-wise.
	\end{proof}
	
\begin{lemma}[EML for pseudoexpectations]
	Let $\{X_{\li}\}_{\li\in L}$ and $\{Y_{\ri}\}_{\ri\in R}$ be two collections of $d$-local functions on $[q]^E$ such that for every $\li\in L$, $X_{\li}(f)$ only depends on $f|_{N_L(\li)}$ and for every $\ri\in R$, $Y_{\ri}(f)$ only depends on $f|_{N_R(\ri)}$.
Then for a $\lambda$-spectral expander, we have
	\[
		\abs{\Ex{\li\sim \ri}{\tildeEx{X_{\li}(\zee) Y_{\ri}(\zee) }} -
                  \Ex{\li,\ri}{\tildeEx{X_{\li}(\zee) Y_{\ri}(\zee)}}} \leq \lambda
                \sqrt{\Ex{\li}{\tildeEx{X_{\li}(\zee)^2}}} \sqrt{\Ex{\ri}
                  {\tildeEx{Y_{\ri}(\zee)^2}}} \mper
	\]
\end{lemma}
\begin{proof}	
Consider the $2n\times 2n$ matrix $M$, with
\[
		M_{ij} = \begin{cases}
			\tildeEx{X_iX_j}, & 1\leq i\leq n, 1\leq j\leq n \\
			\tildeEx{X_iY_{j-n}}, & 1\leq i\leq n, n+1\leq j\leq 2n \\
			\tildeEx{Y_{i-n}X_j}, &n+1\leq i\leq 2n, 1\leq j\leq n\\
			\tildeEx{Y_{i-n}Y_{j-n}} &n+1\leq i\leq 2n,n+1\leq j\leq 2n
		\end{cases}
\]
For any vector $v = (x_1,x_2,\cdots,x_n,y_1,y_2,\cdots,y_n)$, we show that $v^TMv\geq 0$, so that $M$ is PSD.
	\begin{align*}
		v^TMv &~=~ \Ex{i,j}{M_{ij} x_ix_j + M_{i,j+n}x_iy_j + M_{i+n,j} y_ix_j + M_{i+n,j+n} y_iy_j}\\
		&~=~ \Ex{i,j}{\tildeEx{X_iX_j}x_ix_j+\tildeEx{X_iY_j}x_iy_j+\tildeEx{Y_iX_j}y_ix_j+\tildeEx{Y_iY_j}y_iy_j} \\
		&~=~ \Ex{i,j}{\tildeEx{x_i x_j X_iX_j}+\tildeEx{x_iy_jX_iY_j}+\tildeEx{y_ix_jY_iX_j}+\tildeEx{y_iy_jY_iY_j}} \\
		&~=~ \Ex{i,j}{\tildeEx{(x_i X_i + y_i Y_i)(x_jX_j+y_jY_j)}} \\
		&~=~ \tildeEx{\Ex{i,j}{(x_i X_i + y_i Y_i)(x_jX_j+y_jY_j)}} \\
		&~=~ \tildeEx{\Ex{i}{(x_i X_i + y_i Y_i)}^2}  ~\geq~ 0\\
	\end{align*}
Therefore there exist vectors $\{v_{\li}\}_{\li \in L}$ and $\{u_{\ri}\}_{\ri\in R}$ such that 			
\[		
\tildeEx{X_{\li} Y_{\ri}} = \ip{v_{\li}}{u_{\ri}}, \quad
\tildeEx{X_{\li}^2} = \ip{v_{\li}}{v_{\li}}, \quad \text{and} \quad
\tildeEx{Y_{\ri}^2} = \ip{u_{\ri}}{u_{\ri}}
\]
Applying \cref{high_dimensional_eml} to the collection of vectors obtained above, we immediately obtain,
\[
		\abs{\Ex{\li\sim \ri}{\tildeEx{X_{\li}Y_{\ri}}} - \Ex{\li,\ri}{\tildeEx{X_{\li}
                      Y_{\ri}}}} ~\leq~ \lambda \sqrt{\Ex{\li}{\tildeEx{X_{\li}^2}}}
                \sqrt{\Ex{\ri}{\tildeEx{Y_{\ri}^2}}} \qquad \qedhere
\]
\end{proof}

\subsection{Tanner Code}
Suppose we are working with a Tanner code $\TanC$ with inner code $\calC_0$ of distance $\delta_0$, so that the distance of $\TanC$ is at least $\delta_0(\delta_0-\lambda)$. We show that $\eta$-good pseudocodewords satisfy a similar distance property, up to error $\eta$.
\begin{lemma}[Distance of Tanner code]
	\label{lem:distance_from_codeword}
    The distance between an $\eta$-good pseudocodeword $\tildeEx{\cdot}$ and a true codeword $h$ is at least $\delta_0(\delta_0-\lambda) -2\eta \frac{\delta_0}{\delta_0-\lambda}$, or at most $\frac{4\eta^2}{(\delta_0-\lambda)^2} +\frac{\eta (\delta_0+\lambda)}{\delta_0-\lambda}$.
In particular, if $\lambda \leq \delta_0/3$ and $\eta\leq \delta_0^2/9$, then $\dis(h,\tildeEx{\cdot}) \leq 3\eta$ or $\dis(h,\tildeEx{\cdot}) \geq \delta_0(\delta_0-\lambda) - 3\eta$.
\end{lemma}
\begin{proof}
Let $X_{\li}(\zee) \defeq \indi{\zee_{N_L(\li)}\neq h_{N_L(\li)}}$ and $Y_{\ri}(\zee) \defeq
\indi{\zee_{N_R(\ri)}\neq h_{N_R(\ri)}}$, and let $\tau$ denote the quantity
$\sqrt{\Ex{\li}{\tildeEx{X_{\li}(\zee)}}} \cdot \sqrt{\Ex{\ri}{\tildeEx{Y_{\ri}(\zee)}}}$. Then, we have
\begin{align*}
    		\dis\inparen{h,\tildeEx{\cdot}} &~=~ \Ex{e}{\tildeEx{\indi{\zee_e\neq h_e}}} \\
    		&~\leq~ \Ex{\li \sim \ri}{\tildeEx{X_{\li}(\zee) \cdot Y_{\ri}(\zee) }} \\
    		&~\leq~ \Ex{\li , \ri}{\tildeEx{ X_{\li}(\zee) \cdot Y_{\ri}(\zee)}} + \lambda \cdot \sqrt{\Ex{\li}{\tildeEx{X_{\li}(\zee)^2}}} \cdot \sqrt{\Ex{\ri}{\tildeEx{Y_{\ri}(\zee)^2}}}\\
    		&~=~ \Ex{\li , \ri}{\tildeEx{X_{\li}(\zee) \cdot Y_{\ri}(\zee)}} + \lambda \cdot \sqrt{\Ex{\li}{\tildeEx{X_{\li}(\zee)}}} \cdot \sqrt{\Ex{\ri}{\tildeEx{Y_{\ri}(\zee)}}}\\
    		&~=~ \Ex{\li , \ri}{\tildeEx{X_{\li}(\zee) \cdot Y_{\ri}(\zee) }} + \lambda \cdot \tau\\
    		&~=~ \Ex{\li , \ri}{\tildeEx{X_{\li}(\zee)} \cdot \tildeEx{Y_{\ri}(\zee) }} + \eta
           \cdot 1 \cdot 1 +\lambda \cdot \tau \\ &~=~ \tau^2 +\lambda \cdot \tau +\eta
\end{align*}
On the other hand,
\begin{align*}
   		\dis\inparen{h,\tildeEx{\cdot}} &= \Ex{e}{\tildeEx{\indi{\zee_e\neq h_e}}} \\
   		&~=~ \Ex{\li}{\Ex{e\in N_L(\li)}{\tildeEx{\indi{\zee_e\neq h_e}}}} \\
   		&~=~ \Ex{\li}{\tildeEx{\dis\inparen{\zee_{N_L(\li)}, h_{N_L(\li)}}}} \\
   		&~\geq~ \Ex{\li}{\tildeEx{0\cdot \indi{\zee_{N_L(\li)} = h_{N_L(\li)}} + \delta_0 \cdot \indi{\zee_{N_L(\li)} \neq h_{N_L(\li)}}}} \\
   		&~=~ \Ex{\li}{\tildeEx{ \delta_0 \cdot \indi{\zee_{N_L(\li)} \neq h_{N_L(\li)}}}} \\
   		&~=~ \Ex{\li}{\tildeEx{ \delta_0 \cdot X_{\li}(\zee)}} \\
   		&~=~ \delta_0 \cdot  \Ex{\li}{\tildeEx{ X_{\li}(\zee)}} \\
\end{align*}
Likewise, $\dis\inparen{h,\tildeEx{\cdot}} \geq \delta_0 \cdot \Ex{\ri}{\tildeEx{Y_{\ri}(\zee)}}$, and so,
\[
   		\dis\inparen{h,\tildeEx{\cdot}} ~\geq~ \delta_0 \cdot
                \sqrt{\Ex{\li}{\tildeEx{X_{\li}(\zee)}}} \cdot
                \sqrt{\Ex{\ri}{\tildeEx{Y_{\ri}(\zee)}}} ~=~ \delta_0 \cdot \tau
\]
Comparing, we get, $\tau^2 +\lambda \cdot \tau +\eta \geq \delta_0 \cdot \tau$, which means,
\[
\tau ~\geq~ \frac{(\delta_0-\lambda) + \sqrt{(\delta_0-\lambda)^2-4\eta}}{2}
\qquad \text{or} \qquad
\tau ~\leq~ \frac{(\delta_0-\lambda) - \sqrt{(\delta_0-\lambda)^2-4\eta}}{2}
\]
In the first case, we have
\begin{align*}
\dis\inparen{h,\tildeEx{\cdot}} ~\geq~ \delta_0 \cdot \tau 
    		&~\geq~ \delta_0 \cdot \frac{(\delta_0-\lambda) + \sqrt{(\delta_0-\lambda)^2-4\eta}}{2}\\
    		&= \frac{\delta_0(\delta_0-\lambda)}{2} \inparen{ 1+ \sqrt{1-\frac{4\eta}{(\delta_0 - \lambda)^2}} }\\
    		&\geq \frac{\delta_0(\delta_0-\lambda)}{2} \inparen{ 1+ 1-\frac{4\eta}{(\delta_0 - \lambda)^2} }\\
    		&= \delta_0(\delta_0-\lambda) - 2\eta \cdot \frac{\delta_0}{\delta_0-\lambda} \mper
    	\end{align*}
Also, in the second case, we have
\begin{align*}
    		\tau ~\leq~ \frac{(\delta_0-\lambda) - \sqrt{(\delta_0-\lambda)^2-4\eta}}{2} 
    		&~=~ \frac{\delta_0-\lambda}{2} \left( 1- \sqrt{1-\frac{4\eta}{(\delta_0-\lambda)^2}} \right)\\
    		&~\leq~ \frac{\delta_0-\lambda}{2} \left( 1- 1+\frac{4\eta}{(\delta_0-\lambda)^2} \right)\\
    		&~=~ \frac{2\eta}{\delta_0-\lambda} \mcom\\
    	\end{align*}
which gives
\[
\dis(g,\tildeEx{\cdot}) ~\leq~ \tau^2 +\lambda \tau +\eta 
~\leq~ \frac{4\eta^2}{(\delta_0-\lambda)^2} +\frac{2\eta \lambda}{\delta_0-\lambda}+\eta 
~=~ \frac{4\eta^2}{(\delta_0-\lambda)^2} +\frac{\eta (\delta_0+\lambda)}{\delta_0-\lambda}
\qedhere \] 
\end{proof}

\subsection{AEL Code}
Let $\calC_1$ be an outer code on an $(n,d,\lambda)$-expander graph $G(L,R,E)$ and let $\calC_0$ be the inner code. Let $\AELC$ be the code obtained by redistributing symbols along the edges of $G$ and then collecting them on vertices of $R$, as explained in \cref{sec:AEL_prelims}.

Let $\delta_0$ be the distance of $\calC_0$, so that (designed) distance of $\AELC$ is $\delta = \delta_0 - \frac{\lambda}{\delta_1}$. Let $h \in [q_0]^E$ be a codeword in $\AELC$. We show that an $\eta$-good pseudocodeword that has some left-distance from $h$ has a much larger right-distance from $h$.

\begin{lemma}[Distance of AEL Code]\label{lem:AEL_amplification}
	For an $\eta$-good pseudocodeword $\tildeEx{\cdot}$ and a codeword $h \in \AELC$, 
	\[
		\dis^R(\tildeEx{\cdot},h) \geq \delta_0 - \frac{\lambda+\eta}{\dis^L(\tildeEx{\cdot},h)}
	\]
\end{lemma}

\begin{proof}
	We establish upper and lower bounds on $\dis(\tildeEx{\cdot},h)$.
	\begin{align*}
		\dis(\tildeEx{\cdot},h) &= \Ex{e}{\tildeEx{\indi{\zee_e\neq h_e}}} \\
		&~\geq~ \Ex{\li\in L}{\tildeEx{0\cdot \indi{\zee_{N_L(\li)} = h_{N_L(\li)}} + \delta_0 \cdot \indi{\zee_{N_L(\li)}\neq h_{N_L(\li)}}}} \\
		&~=~ \delta_0 \Ex{\li\in L}{\tildeEx{\indi{\zee_{N_L(\li)}\neq h_{N_L(\li)}}}} \\
		&~=~ \delta_0 \dis^L(\tildeEx{\cdot},h)
	\end{align*}
For the upper bound, we again rely on the expander mixing lemma:
\begin{align*}
		\dis(\tildeEx{\cdot},h) &~=~ \Ex{e}{\tildeEx{\indi{\zee_e\neq h_e}}} \\
		&~\leq~ \Ex{\li\sim \ri}{\tildeEx{\indi{\zee_{N_L(\li)}\neq h_{N_L(\li)}} \indi{\zee_{N_R(\ri)}\neq h_{N_R(\ri)}}}} \\
		&~\leq~ \Ex{\li ,\ri}{\tildeEx{\indi{\zee_{N_L(\li)}\neq h_{N_L(\li)}} \indi{\zee_{N_R(\ri)}\neq h_{N_R(\ri)}}}} +\lambda\\
		&~\leq~ \Ex{\li ,\ri}{\tildeEx{\indi{\zee_{N_L(\li)}\neq h_{N_L(\li)}}} \tildeEx{\indi{\zee_{N_R(\ri)}\neq h_{N_R(\ri)}}}} +\lambda +\eta\\
		&~=~ \Ex{\li}{\tildeEx{\indi{\zee_{N_L(\li)}\neq h_{N_L(\li)}}}} \Ex{\ri}{\tildeEx{\indi{\zee_{N_R(\ri)}\neq h_{N_R(\ri)}}}} +\lambda +\eta\\
		&~=~ \dis^L(\tildeEx{\cdot},h) \cdot \dis^R(\tildeEx{\cdot},h) +\lambda+\eta
	\end{align*}
Dividing the two bounds by $\dis^L(\tildeEx{\cdot},h)$ and rearranging, we finally get,
\[
\dis^R(\tildeEx{\cdot},h) ~\geq~ \delta_0 - \frac{\lambda+\eta}{\dis^L(\tildeEx{\cdot},h)} \mper
\qquad \qquad \qquad \qquad \qedhere
\]
\end{proof}


\section{Correlation Reduction via Conditioning}
\label{sec:conditioning}
We will use the following claim from \cite{BRS11} (see Lemma 5.2 there) that says that if $\zee_S$ and $\zee_T$ have a large covariance, then conditioning on $\zee_T$ reduces the variance of $\zee_S$ significantly.
\begin{lemma}
	Let $\tildeEx{\cdot}$ be a pseudoexpectation operator of SoS-degree $t$ with associated pseudocovariance and pseudovariance operators. Assume $S,T$ are sets such that $|S|+|T|\leq t/2$, then,
	\[
		\tildeVar{\zee_S | \zee_T} \leq \tildeVar{\zee_S} - \frac{1}{q^{|T|}} \sum_{\alpha \in [q]^S, \beta\in [q]^T} \frac{(\tildeCov{\zee_{S,\alpha}}{\zee_{T,\beta}})^2}{\tildeVar{\zee_{T,\beta}}}
	\]
\end{lemma}
In particular, observe that pseudovariances are non-increasing under conditioning.
The next lemma shows that if the average covariance across all pairs $(\li,\ri)$ is $\eta$, then conditioning on a random vertex in $R$ will reduce the average variance in $L$ in expectation by $\Omega(\eta^2)$. Then, \cref{lem:low_covariance_solution} will use that this cannot happen more than $\calO(1/\eta^2)$ times, and then we must end up with a conditioned pseudoexpectation operator which has low average covariance, that is, it is $\eta$-good.
\begin{lemma}\label{lem:conditioning_reduces_variance}
    Let $\eta < \Ex{\li,\ri}{\tildeCov{\zee_{N_L(\li)}}{\zee_{N_R(\ri)}}} $. Then,
    \[
        \Ex{\ri \in R}{\Ex{\li}{ \tildeVar{\zee_{N_L(\li)} \vert \zee_{N_R(\ri)}}}} < \Ex{\li}{ \tildeVar{\zee_{N_L(\li)}}} - \frac{1}{q^{2d}}{\eta^2}
    \]
\end{lemma}
\begin{proof}

    \begin{align*}
        \Ex{\ri \in R}{\Ex{\li}{\tildeVar{\zee_{N_L(\li)}|\zee_{N_R(\ri)}}}} &= \Ex{\li,\ri}{\tildeVar{\zee_{N_L(\li)}|\zee_{N_R(\ri)}}}\\
        &\leq \Ex{\li,\ri}{\tildeVar{\zee_{N_L(\li)}} - \frac{1}{q^d} \sum_{\alpha, \beta} \frac{(\tildeCov{\zee_{N_L(\li),\alpha}}{\zee_{N_R(\ri),\beta}})^2}{\tildeVar{\zee_{N_R(\ri),\beta}}}} \\
        &\leq \Ex{\li,\ri}{\tildeVar{\zee_{N_L(\li)}} - \frac{1}{q^d} \sum_{\alpha, \beta} \left( \tildeCov{\zee_{N_L(\li),\alpha}}{\zee_{N_R(\ri),\beta}} \right)^2} \\
        &\leq \Ex{\li,\ri}{\tildeVar{\zee_{N_L(\li)}} - \frac{1}{q^{3d}} \left( \sum_{\alpha, \beta} \abs{ \tildeCov{\zee_{N_L(\li),\alpha}}{\zee_{N_R(\ri),\beta}}}\right)^2} \\
        &= \Ex{\li}{\tildeVar{\zee_{N_L(\li)}}} - \frac{1}{q^{3d}} \Ex{\li,\ri}{\left( \tildeCov{\zee_{N_L(\li)}}{\zee_{N_R(\ri)}} \right)^2} \\
        &\leq \Ex{\li}{\tildeVar{\zee_{N_L(\li)}}} - \frac{1}{q^{3d}} \left( \Ex{\li,\ri}{ \tildeCov{\zee_{N_L(\li)}}{\zee_{N_R(\ri)}} } \right)^2 
\qedhere    \end{align*}
\end{proof}
\begin{lemma}\label{lem:low_covariance_solution}
	Let $\eta>0$ be arbitrarily small. Given any SoS solution $\tildeEx{\cdot}$ of degree $\geq 2d\left(\frac{q^{3d}}{\eta^2}+1\right)$, there exists a number $k^* \leq q^{3d}/\eta^2 $ such that 
	\[
		\Ex{v_1,v_2,\cdots,v_{k^*}}{\Ex{\li,\ri}{\tildecov[\zee_{N_L(\li)},\zee_{N_R(\ri)} \vert \zee_{N_R(v_1)},\zee_{N_R(v_2)},\cdots,\zee_{N_R(v_{k^*})}]}} \leq \eta
	\]
\end{lemma}
\begin{proof}
Consider $\Phi_k(\tildeEx{\cdot}) = \Ex{v_1,v_2,\cdots,v_k}{\Ex{\li}{\tildeVar{\zee_{N_L(\li)}|\zee_{N_R(v_1)},\zee_{N_R(v_2)},\cdots,\zee_{N_R(v_k)}}}}$. We know that 
\[
1~\geq~ \Phi_0 ~\geq~ \Phi_1~\geq~ \cdots ~\geq~ \Phi_{q^{3d}/\eta^2} ~\geq~ 0
\]
so there exists a $k^* \leq q^{3d}/\eta^2$ such that $\Phi_{k^*}(\tildeEx{\cdot}) - \Phi_{k^*+1}(\tildeEx{\cdot}) \leq \frac{1}{q^{3d}/\eta^2} = \eta^2/q^{3d}$.
By contrapositive of \cref{lem:conditioning_reduces_variance}, this means that 
	\[
		\Ex{v_1,v_2,\cdots,v_{k^*}}{\Ex{\li,\ri}{\tildecov[\zee_{N_L(\li)},\zee_{N_R(\ri)} \vert \zee_{N_R(v_1)},\zee_{N_R(v_2)},\cdots,\zee_{N_R(v_{k^*})}]}} \leq \eta
\qedhere	
	\]
\end{proof}


\section{List Decoding up to Johnson Bound}
\label{sec:decoding}
In this section, we combine different pieces of the proof to give list decoding algorithms up to the Johnson bound. Note that in both Tanner and AEL cases, we reduce to unique decoding of either the same code or the base code, and this unique decoding needs to be done from pseudocodewords instead of codewords. We will handle this slight strengthening of unique decoding via randomized rounding in \cref{sec:decoding_from_fractional}.
\subsection{Tanner code}\label{sec:list_decoding_tanner}
Let $\TanC$ be a Tanner code on an $(n,d,\lambda)$-expander graph $G(L,R,E)$, with $\calC_0$ as the inner code. Let $\delta_0$ be the distance of $\calC_0$, so that (designed) distance of $\TanC$ is $\delta = \delta_0(\delta_0-\lambda)$. Assume $\lambda \leq \delta_0/3$. Given $g\in [q]^E$, we wish to recover the list $\calL(g,\calJ(\delta)-\eps)$. As $\eps \rightarrow 0$, the decoding radius gets arbitrarily close to the Johnson bound.

\begin{theorem}[List decoding Tanner codes]\label{thm:tanner-decoding}
There is a deterministic algorithm based on $\calO_{q,d}(1/\eps^4)$ levels of the SoS-hierarchy that given $g$ runs in time $n^{O_{q,d}(1/\eps^4)}$ time and computes the list $\calL(g,\calJ(\delta)-\eps))$.
\end{theorem}
\begin{proof}
We apply the algorithmic covering \cref{lem:cover_for_list_tanner} to obtain a pseudocodeword $\tildeEx{\cdot}$ of SoS-degree $t \geq 2d\left( \frac{q^{3d}}{\eta^2} + 2\right)$ such that for any $h \in \calL(g,\calJ(\delta)-\eps)$, we know that $\dis(\tildeEx{\cdot},h) \leq \delta - \eps_2$ for $\eps_2 = 2\eps \cdot \sqrt{1-\frac{q}{q-1} \delta} \geq \Omega(\eps)$. We will choose $\eta$ later, and note that the choice of $\eta$ does not change $\eps_2$.
Henceforth, we fix an $h\in \calL((g,\calJ(\delta) - \eps)$, so that $\dis(\tildeEx{\cdot},h) \leq \delta -\eps_2$. Our goal is to recover $h$.

Pick a random $k \in \{1,\cdots,\ceil{\frac{q^{3d}}{\eta^2}}\}$. From \cref{lem:low_covariance_solution}, we know that with probability at least $\frac{\eta^2}{q^{3d}}$, 
\begin{align}\label{eqn:random_conditioning}
\Ex{v_1,v_2,\cdots ,v_k}{\Ex{\li,\ri}{\tildecov[\zee_{N_L(\li)},\zee_{N_R(\ri)} |
  \zee_{N_R(v_1)},\zee_{N_R(v_2)},\cdots ,\zee_{N_R(v_k)}]}} \leq \eta \mper
\end{align}
	We assume that we found a $k$ such that \cref{eqn:random_conditioning} holds. Let $V$ be the (random) set of $k$ vertices we condition on, that is, $V = \{v_1,v_2,\cdots,v_k\}$, and let $N_R(V) = \cup_{v\in V} N_R(v)$. Then,
	\[
		\Ex{V\sub R, |V|=k}{\Ex{\li,\ri}{\tildecov[\zee_{N_L(\li)},\zee_{N_R(\ri)} | \zee_{N_R(V)}]}} \leq \eta
	\]
	More explicitly, conditioning on $N_R(V)$ involves sampling an assignment for $N_R(V)$ according to the local distribution of $\zee_{N_R(V)}$. Let this random assignment be $\beta$, and we get
	\[
		\Ex{\substack{V \sub R,|V|=k \\ \beta \sim \zee_{N_R(V)}}}{\Ex{\li,\ri}{\tildecov[\zee_{N_L(\li)},\zee_{N_R(\ri)} | \zee_{N_R(V)}=\beta]}} \leq \eta
	\]
By Markov's inequality,
	\[
	\Pr{\substack{V \sub R,|V|=k \\ \beta \sim \zee_{N_R(V)}}}{\Ex{\li,\ri}{\tildecov[\zee_{N_L(\li)},\zee_{N_R(\ri)} | \zee_{N_R(V)}=\beta]}> \frac{\eps_2}{15}} \leq \frac{15\eta}{\eps_2}
	\]
By choosing $\eta=\eps_2^2/60\delta$, we get
\begin{equation}\label{eqn:low_covariance}
	\Pr{\substack{V \sub R,|V|=k \\ \beta \sim \zee_{N_R(V)}}}{\Ex{\li,\ri}{\tildecov[\zee_{N_L(\li)},\zee_{N_R(\ri)} | \zee_{N_R(V)}=\beta]} \leq \frac{\eps_2}{15}} \geq 1- \eps_2/4\delta
\end{equation}
For some $V,\beta$ such that \cref{eqn:low_covariance} holds, we will be using $\tildeEx{~ \cdot~ |
  \zee_{N_R(V)}=\beta}$ as an $\eta$-good pseudocodeword. Using $\eps_2 < \delta$, note that 
\begin{align*}
		\frac{\eps_2}{15} < \frac{\delta}{15} = \frac{\delta_0\cdot (\delta_0-\lambda)}{15}< \frac{\delta_0^2}{9}
\end{align*}
and $\lambda\leq \delta_0/3$ so that the conditions of \cref{lem:distance_from_codeword} are satisfied. 
	 
	 For this $\eta$-good pseudocodeword, we need to argue that it is still close to $h$ that we are trying to find. This is easy to ensure in expectation, and we again appeal to Markov's inequality to say that it also holds with significant probability, up to some loss in distance. By the law of total expectation, for any $V \sub R$,
\[
		\Ex{\beta \sim \zee_{N_R(V)}}{\dis\inparen{\tildeEx{~\cdot~ | \zee_{N_R(V)}=\beta},h}} = \dis(\tildeEx{\cdot},h) \leq \delta-\eps_2
\]
Averaging over all $V\sub R$ of size $k$,
\begin{equation}
		\Ex{\substack{V \sub R,|V|=k \\ \beta \sim \zee_{N_R(V)}}}{\dis\inparen{\tildeEx{\ \cdot\ | \zee_{N_R(V)}=\beta},h}} \leq \delta-\eps_2
\end{equation}
Again, we claim via Markov's inequality that a significant fraction of conditionings must end up being not too far from $f$.
\begin{equation}\label{eqn:good_agreement}
		\Pr{\substack{V \sub R,|V|=k \\ \beta \sim \zee_{N_R(V)}}}{\dis\inparen{ \tildeEx{\ \cdot\ | \zee_{N_R(V)}=\beta},h} \leq \delta-\frac{\eps_2}{2}} \geq \frac{\eps_2/2}{\delta-\eps_2+\frac{\eps_2}{2}} \geq \frac{\eps_2}{2\delta}
\end{equation}
Henceforth, we fix a conditioning $(V,\beta)$ with $\beta\in [q]^{N_R(V)}$ such that events in both \cref{eqn:low_covariance} and \cref{eqn:good_agreement} happen. Note that by a union bound, a random $(V,\beta)$ has this property with probability at least $\eps_2/4\delta$. Fix such a conditioning, and let the conditioned pseudoexpectation be \[ \dupPE{\cdot} = \tildeEx{~\cdot~ | \zee_{N_R(V) = \beta}}\] and the corresponding covariance operator be $\dupCov[\cdot]$ Note that the degree of $\dupPE{\cdot}$ is $t-2d\cdot k \geq 4d$. From definition, we know that
\begin{align}
\dis(\dupPE{\cdot},h) \leq \delta-\frac{\eps_2}{2} \label{eqn:good_agreement_new_operator}\\
\Ex{\li,\ri}{\dupCov[ \zee_{N_L(\li)}, \zee_{N_R(\ri)}]} \leq \frac{\eps_2}{15}\label{eqn:low_covariance_new_operator}
\end{align}
From \cref{eqn:low_covariance_new_operator} and \cref{lem:distance_from_codeword}, we know that $\dis(\dupPE{\cdot},h) \leq \eps_2/5$ or $\dis(\dupPE{\cdot},h) \geq \delta - \eps_2/5$. The latter is impossible because of \cref{eqn:good_agreement_new_operator}, and so we must have $\dis(\dupPE{\cdot},h) \leq \eps_2/5 < \delta/5$.
	
Finally, we use \cref{lem:unique_decoding} to recover $h$ using $\dupPE{\cdot}$ with probability at least $1/5$.
The algorithm succeeds if 
\begin{enumerate}
\item a $k$ is picked so that \cref{eqn:random_conditioning} holds, 
\item $(V,\beta)$ is picked so that events in both \cref{eqn:low_covariance} and \cref{eqn:good_agreement} happen,
\item the call to \cref{lem:unique_decoding} succeeds.
\end{enumerate}
The success probability is therefore at least
\[
\Pr{\text{success}} ~\geq~ \frac{\eta^2}{q^{2d}} \cdot \frac{\eps_2}{4\delta} \cdot \frac{1}{5} ~\geq~
\Omega \left( \frac{\eps_2^5}{\delta^3\cdot q^{2d}} \right) ~=~ \Omega_{q,d}(\eps^5) \mper
\]
We have shown that for any $h\in \calL((g,\calJ(\delta)-\eps)$, the algorithm above outputs $h$ with probability at least $\Omega_{q,d}(\eps^5)$. Note that this implicitly proves an upper bound on the list of $\calO_{q,d}(1/\eps^5)$.
	
Therefore, the random choices that the algorithm makes lead it to different elements of the list. We next argue that we can derandomize all random choices in the algorithm, so that all elements of the list can be found with a deterministic algorithm.
	\begin{enumerate}
		\item For the random choice of $k$, we can try out all possible $q^{3d}/\eta^2$ values for $k$. 
		\item For random $(V,\beta)$, we can again try out all possible values, which are at most $n^k\cdot 2^k \leq n^{\calO_{q,d}(1/\eps^4)}$ in number.
		\item \cref{lem:unique_decoding} can be derandomized using a standard threshold rounding argument, as argued in \cref{lem:derandomized_decoding_from_distributions}.
	\end{enumerate}
Thus, the final algorithm starts with an empty list and goes over all the deterministic choices above. Every $h\in \calL(g,\calJ(\delta)-\eps)$ will be discovered in at least one of these deterministic steps, and we can efficiently check whether $\Delta(g,h) <\calJ(\delta)-\eps$. If yes, $h$ is added to the output list.
\end{proof}

\subsection{AEL Code}\label{sec:list_decoding_ael}
Let $\AELC$ be an AEL code determined by an $(n,d,\lambda)$-expander graph $G(L,R,E)$, an inner code $\calC_0$ of distance $\delta_0$, rate $r_0$, alphabet size $q_0$, and an outer code $\calC_1$ of distance $\delta_1$, rate $r_1$ and alphabet size $q_1 = |\calC_0|$. The code $\AELC$ is of alphabet size $q_0^d$, rate $r_0r_1$ and (designed) distance $\delta_0-\frac{\lambda}{\delta_1}$. 
\begin{theorem}[List Decoding AEL codes]\label{thm:list_decoding_ael}
	Suppose the code $\calC_1$ can be efficiently unique-decoded from radius $\delta_{dec}$. Assume $\lambda \leq \kappa \cdot \delta_{dec} \leq \kappa \cdot \delta_1$, so that the distance of $\AELC$ is at least $\delta_0-\kappa$. Then for any $\eps > 0$, the code $\AELC$ can be list decoded from a radius of $\calJ(\delta_0-\kappa)-\eps$ by using $\calO_{q,d,\delta_{dec}}(1/\eps^4)$ levels of SoS-hierarchy, in time $n^{\calO_{q,d,\delta_{dec}}(1/\eps^4)}$.
\end{theorem}
\begin{proof}
Let $g\in [q_0^d]^R$ be a received word. Recall that the distance of an AEL codeword $h$ with $g$ is given by $\dis^R(g,h) = \Ex{\ri\in R}{\indi{g(\ri)\neq h_{N_R(\ri)}}}$.
	
We again start by applying the algorithmic covering \cref{lem:cover_for_list_ael} to get a pseudocodeword $\tildeEx{\cdot}$ of SoS-degree $t\geq 2d(\frac{q^{3d}}{\eta^2}+2)$ such that for every $h \in \calL(g,\calJ(\delta_0-\kappa)-\eps)$,
\[
\Ex{\ri}{\tildeEx{\indi{\zee_{N_R(\ri)} \neq h_{N_R(\ri)}}}} = \dis^R(\tildeEx{\cdot},h) \leq (\delta_0-\kappa)-\eps_2
\]
Here $\eta>0$ is a small constant to be chosen later, and $\eps_2 = 2\eps \sqrt{1-\frac{q^d}{q^d-1}(\delta_0-\kappa)} \geq \Omega(\eps)$.
	Henceforth, we fix an $h\in \calL\left((g,\calJ(\delta_0-\kappa) - \eps \right)$, so that $\dis^R(\tildeEx{\cdot},h) \leq \delta_0-\kappa -\eps_2$. Our goal is to recover $h$.
	
	Pick a random $k \in \{1,\cdots,\ceil{\frac{q^{3d}}{\eta^2}}\}$. From \cref{lem:low_covariance_solution}, we know that with probability at least $\frac{\eta^2}{q^{3d}}$,
	\[
		\Ex{v_1,v_2,\cdots ,v_k}{\Ex{\li,\ri}{\tildecov[\zee_{N_L(\li)},\zee_{N_R(\ri)} | \zee_{N_R(v_1)},\zee_{N_R(v_2)},\cdots ,\zee_{N_R(v_k)}]}} \leq \eta
	\]
We assume that we found a $k$ such that \cref{eqn:random_conditioning} holds. Let $V$ be the (random) set of $k$ vertices we condition on, that is, $V = \{v_1,v_2,\cdots,v_k\}$, and let $N_R(V) = \cup_{v\in V} N_R(v)$. Then,
\[
\Ex{V\sub R, |V|=k}{\Ex{\li,\ri}{\tildecov[\zee_{N_L(\li)},\zee_{N_R(\ri)} | \zee_{N_R(V)}]}} \leq \eta
\]
More explicitly, conditioning on $N_R(V)$ involves sampling an assignment for $N_R(V)$ according to the local distribution of $\zee_{N_R(V)}$. Let this random assignment be $\beta$, and we get
\[
\Ex{\substack{V \sub R,|V|=k \\ \beta \sim
    \zee_{N_R(V)}}}{\Ex{\li,\ri}{\tildecov[\zee_{N_L(\li)},\zee_{N_R(\ri)} | \zee_{N_R(V)}=\beta]}}
\leq \eta \mper
\]
By Markov's inequality, and by choosing $\eta = \frac{\eps_2^2\delta_{dec}}{16(\delta_0-\kappa)}$,
	\begin{gather}
	\Pr{\substack{V \sub R,|V|=k \\ \beta \sim \zee_{N_R(V)}}}{\Ex{\li,\ri}{\tildecov[\zee_{N_L(\li)},\zee_{N_R(\ri)} | \zee_{N_R(V)}=\beta]}> \frac{\delta_{dec} \cdot \eps_2}{4}} \leq \frac{4\eta}{\delta_{dec}\cdot\eps_2} \leq \frac{\eps_2}{4(\delta_0-\kappa)} \\
	\Pr{\substack{V \sub R,|V|=k \\ \beta \sim \zee_{N_R(V)}}}{\Ex{\li,\ri}{\tildecov[\zee_{N_L(\li)},\zee_{N_R(\ri)} | \zee_{N_R(V)}=\beta]} \leq \frac{\delta_{dec} \cdot \eps_2}{4}} \geq 1-\frac{\eps_2}{4(\delta_0-\kappa)}\label{eqn:eta_good_ael}
	\end{gather}
	As in the Tanner case, we next claim that the distance is preserved with significant probability when conditioning randomly. Let $h \in \calL(g,\calJ(\delta_0-\kappa)-\eps)$ so that $\dis^R(\tildeEx{\cdot},h) < \delta_0-\kappa -\eps_2$. Using a similar argument as in the Tanner case,
	\begin{equation}
		\Ex{(V,\beta)}{\dis^R\left( \tildeEx{\ \cdot\ | \zee_{N_R(V)}=\beta},h\right)} \leq (\delta_0-\kappa)-\eps_2
	\end{equation}
which allows us to claim via Markov's inequality that
\begin{align}\label{eqn:good_agreement_ael}
		\Pr{(V,\beta)}{\dis^R \left( \tildeEx{\ \cdot\ | \zee_{N_R(V)}=\beta},h \right) \leq (\delta_0-\kappa)-\frac{\eps_2}{2}} ~\geq~ \frac{\eps_2/2}{(\delta_0-\kappa)-\eps_2+\eps_2/2} 
		~\geq~ \frac{\eps_2}{2(\delta_0-\kappa)}
	\end{align}
Again, let $(V,\beta)$ be a conditioning such that events in both \cref{eqn:eta_good_ael} and \cref{eqn:good_agreement_ael} hold (which happens with probability at least $\frac{\eps_2}{4(\delta_0-\kappa)}$). Let $\dupPE{\cdot} = \tildeEx{\cdot | \zee_{N_R(V)} =\beta}$ be an $\eta$-good pseudocodeword that satisfies \cref{eqn:good_agreement_ael}.
This means
\[
\delta_0 - \frac{\lambda+\eta}{\dis^L(\dupPE{\cdot},h)} ~\leq~ \dis^R(\dupPE{\cdot},h) ~\leq~
(\delta_0-\kappa)-\eps_2/2 
\]
Rearranging, we get that
\[
\kappa+\eps_2/2 ~\leq~ \frac{\lambda+\eta}{\dis^L(\dupPE{\cdot},h)}
~\leq~ \frac{\lambda+\delta_{dec}\cdot\eps_2/4}{\dis^L(\dupPE{\cdot},h)}\\
~\leq~ \frac{\kappa \cdot \delta_{dec}+\delta_{dec}\cdot\eps_2/4}{\dis^L(\dupPE{\cdot},h)} \mcom
\]
which gives the required bound on $\dis^L(\dupPE{\cdot},h)$ as
\[
\dis^L(\dupPE{\cdot},h) ~\leq~ \delta_{dec} - \frac{\delta_{dec} \eps_2/4}{\kappa+\eps_2/2} 
		~\leq~ \delta_{dec} - \frac{\delta_{dec} \eps_2}{4\delta_0}
\]
Finally, we use \cref{lem:unique_decoding_ael} to find $h$ using $\dupPE{\cdot}$ with probability at least $\eps_2/4\delta_0$.
The final success probability is at least
\[
\frac{\eta^2}{q^{3d}} \cdot \frac{\eps_2}{4(\delta_0-\kappa)} \cdot \frac{\eps_2}{4\delta_0} ~\geq~
\Omega\inparen{\frac{\eps_2^6\delta_{dec}^2}{q^{3d}\delta_0^4}} ~\geq~
\Omega_{q,d,\delta_{dec}}(\eps^6) \mper
	\]
	Just as in the case of Tanner codes, this algorithm can be derandomized by trying out all possible random choices made by the algorithm, to give a deterministic algorithm that recovers the list.
\end{proof}
Note that while the theorem above deals with list decoding, it can be easily adapted for list recovery by replacing the use of \cref{lem:cover_for_list_ael} by \cref{lem:cover_for_list_recovery_ael}.
Next, we use the AEL amplification scheme to construct near-MDS codes list decodable up to the Johnson bound.

\begin{theorem}\label{thm:near_mds_main}
	For any $\nfrac{1}{2} >\eps_1, \eps_2>0$, there is an infinite family of codes $\calC$ of blocklength $n$ with the following properties:
	\begin{enumerate}[(i)]
		\item The rate of the code is $\rho$ and distance is at least $1-\rho-\eps_1$.
		\item The code is over an alphabet of size $2^{\calO(\eps_1^{-6}\log(1/\eps_1))}$.
		\item The code can be list decoded from radius $\calJ(1-\rho - \eps_1)-\eps_2$ in time $n^{\calO_{\eps_1}(1/\eps_2^4)}$.
	\end{enumerate}
\end{theorem}

\begin{proof}
We sketch how to instantiate \cref{thm:list_decoding_ael} to obtain such codes.

Suppose we are working with $(n,d,\lambda)$-expander.

Choose the inner code $\calC_0$ to be a Reed Solomon code of rate $\rho_0$, distance $1-\rho_0$ and alphabet size $q_0=d$ (or any MDS code). Choose the outer code $\calC_1$ over alphabet of size $|\calC_0| = d^{\rho_0\cdot d}$ to have rate $1-\eps_1$ that can be unique decoded from radius $\delta_{dec}=\Omega(\eps_1^2)$, such as the one constructed in \cite{GI05}.

Let $\lambda = \kappa\delta_{dec}$ with $\kappa =\eps_1$, so that $\lambda \leq \Theta(\eps_1^3)$ and $d=\Theta(1/\eps_1^6)$.

The rate of the final AEL code is $\rho \defeq (1-\eps_1) \rho_0$, and the distance is at least
\begin{align*}
	(1-\rho_0) - \frac{\lambda}{\delta_1} &~\geq~ 1-\frac{\rho}{1-\eps_1} - \frac{\lambda}{\delta_{dec}} \\
	&~\geq~1-\rho -2\eps_1 \rho -\eps_1 \\
	&~\geq~1-\rho - 3\eps_1
\end{align*}

The alphabet size is $q_0^d = d^d  = 2^{ \calO\inparen{\eps_1^{-6}\log(1/\eps_1)}}$.

For list decodability, we use \cref{thm:list_decoding_ael} to claim that the above code can be list decoded from $\calJ(1-\rho-3\eps_1)-\eps_2$ in time $n^{\calO_{q_0,d,\delta_{dec}}(1/\eps_2^4)} = n^{\calO_{\eps_1}(1/\eps_2^4)}$.

Replace $\eps_1$ by $\eps_1/3$ to get the final result.
\end{proof}

Note that we can also deal with $\calC_1$ that can be decoded from smaller radius like $\calO(\eps_1^3)$, by suitably adjusting $\lambda$ and paying the cost in alphabet size. Above, we have not chosen parameters optimally to keep the exposition simple. The alphabet size in the code constructed in \cite{GI05} was smaller than what we ask for here, but alphabet size can always be increased while preserving rate, distance and unique decoding radius by folding multiple symbols together into one. This looks like multiple symbols of the outer code being assigned to the same left node in the AEL construction.

\subsection{Decoding from fractional vectors}\label{sec:decoding_from_fractional}
This section has auxiliary claims needed to finish the list decoding algorithm proof. In the Tanner code case, we reduce to unique decoding of the same code from an arbitrarily small radius. In AEL, we reduce to unique decoding of the base (outer) code. In both these cases, we have some $\eps$ slack, that is sufficient for randomized rounding to produce a (corrupted) word within the unique decoding radius.
\begin{lemma}\label{lem:decoding_from_distributions}
	Let $\calC$ be an $[n,\delta,\rho]_q$ code, which is unique decodable from distance $\delta_{dec}\leq \delta/2$ in time $\calT(n)$. Given a collection of distributions $\calD = \inbraces{\calD_i}_{i\in [n]}$, each of them supported on $[q]$, there is a unique codeword $h$ that satisfies
	\[
		\Ex{i}{\Ex{j\sim \calD_i}{\indi{h_i \neq j}}} \le \delta_{dec}-\eps
	\]
	This codeword $h$ can be found in time $\calO(qn)+\calT(n)$ with probability at least $\eps/\delta_{dec}$.
\end{lemma}
\begin{proof}
First, we show uniqueness of $h$. Let $h, h'$ be two codewords in $\calC$ such that 
\[		
\Ex{i}{\Ex{j\sim \calD_i}{\indi{h_i \neq j}}} ~\le~ \delta_{dec}-\eps 
\qquad \text{and} \qquad 
\Ex{i}{\Ex{j\sim \calD_i}{\indi{h'_i \neq j}}} ~\le~ \delta_{dec}-\eps
\]
For any $j\in [q]$, we have $\indi{h_i\neq h'_i} \leq \indi{h_i\neq j}+\indi{h'_i\neq j}$. Therefore,
\begin{align*}
\dis(h,h') ~=~ \Ex{i}{\indi{h_i\neq h'_i}} &~\leq~ \Ex{i}{\Ex{j\sim \calD_i}{\indi{h_i\neq j}+\indi{h'_i\neq j}}} \\
		&~=~ \Ex{i}{\Ex{j\sim \calD_i}{\indi{h_i\neq j}}}+\Ex{i}{\Ex{j\in \calD_i}{\indi{h'_i\neq j}}} 
		~\leq~ 2\delta_{dec}-2\eps < \delta \mper
\end{align*}
By distance property of the code, this means that $\dis(h,h') = 0$, or that $h=h'$.

The algorithm to find $h$ is to independently sample from every distribution to get a random $g \in [q]^n$, and then issue a unique decoding call from $g$. We show that with significant probability, $h$ lies in the $\delta_{dec}$ radius ball around $g$, which will show that algorithm succeeds with that probability.
\begin{align*}
\Ex{g}{\dis(g,h)} ~=~ \Ex{g}{\Ex{i}{\indi{g_i\neq h_i}}} 
		~=~ \Ex{i}{\Ex{g}{\indi{g_i\neq h_i}}} 
		~=~ \Ex{i}{\Ex{g_i \sim \calD_i}{\indi{g_i\neq h_i}}} 
		~\leq~ \delta_{dec}-\eps \mper
	\end{align*}
Thus, by Markov's inequality, we have, $\Pr{g}{\dis(g,h) \leq \delta_{dec}} \geq
\frac{\eps}{\delta_{dec}}$, which proves the claim.
\end{proof}
\begin{remark}
		The success probability in the \cref{lem:decoding_from_distributions} can be amplified by repeated sampling. Moreover, the fact that we reduce to the unique decoding algorithm of $\calC$ is not important, as it is also possible to use a list decoding algorithm for $\calC$ from distance $\delta_{dec}$ to find $h$, as long as the sampled $g$ satisfies $\delta(g,h) < \delta_{dec}$. In that case, we output a random element of list obtained, which incurs an additional loss of $1/L$ factor, where $L$ is the list size guaranteed by list decoding algorithm for $\calC$ up to $\delta_{dec}$.
	\end{remark}
	
	\begin{remark}
		As shown in \cref{lem:derandomized_decoding_from_distributions}, this argument can be derandomized using threshold rounding. The use of \cref{lem:decoding_from_distributions} in the next two lemmas can therefore be replaced by \cref{lem:derandomized_decoding_from_distributions}.
	\end{remark}
	
	Next, we use pseudocodewords that lie in the unique decoding ball to construct the collection of distributions needed by \cref{lem:decoding_from_distributions}, for the Tanner and AEL cases.
	\begin{lemma}[Unique decoding from Tanner pseudocodewords]\label{lem:unique_decoding}
		Let $\TanC$ be a code as in \cref{sec:list_decoding_tanner}, with distance $\delta = \delta_0(\delta_0-\lambda)$ and $\lambda\leq \delta_0/3$, and in particular it can be unique decoded from radius $\delta/4$ in time $\calO(|E|)$. Given a Tanner pseudocodeword $\tildeEx{\cdot}$ such that $\dis(\tildeEx{\cdot},h) <\delta/5$, we can find $h$ in time $\calO(|E|)$ with probability at least $1/5$.
	\end{lemma}
	\begin{proof}
		The pseudocodeword gives a collection of distributions $\calD = \inbraces{\calD_e}_{e\in E}$, each distribution over $[q]$. The $e^{th}$ distribution $\calD_e$ gives a weight of $\tildeEx{\indi{\zee_e=j}}$ to the value $j\in [q]$.
		
		The distance $\dis(\tildeEx{\cdot},h)$ translates to
		\begin{align*}
			\dis(\tildeEx{\cdot},h) &= \Ex{e}{\tildeEx{\indi{\zee_e \neq h_e}}} \\
			&= \Ex{e}{\tildeEx{ \sum_{j\in [q]}  \indi{\zee_e =j }\indi{h_e \neq j}}} \\
			&=\Ex{e}{ \sum_{j\in [q]} \indi{h_e \neq j} \tildeEx{  \indi{\zee_e =j }}} \\
			&=\Ex{e}{ \Ex{j\sim \calD_e}{ \indi{h_e \neq j}}} \\
		\end{align*}
		We can therefore use this collection of distributions to obtain a codeword $h\in \TanC$ via \cref{lem:decoding_from_distributions} with $\delta_{dec} = \delta/4$ and $\eps = \delta/20$ such that $\dis(\tildeEx{\cdot},h) <\delta/5$.
	\end{proof}

	Next, we use similar ideas to round and decode from AEL pseudocodewords. We borrow the terminology used for AEL codes from \cref{sec:list_decoding_ael}.
	\begin{lemma}[Unique decoding from AEL pseudocodewords]\label{lem:unique_decoding_ael}
		Let $\AELC$ be a code as in \cref{sec:list_decoding_ael}, with $\lambda\leq \kappa\cdot \delta_{dec}$  and distance at least $\delta_0 - \kappa$. Assume that the outer code $\calC_1$ can be unique decoded from radius $\delta_{dec}$ in time $\calT(n)$. Given an AEL pseudocodeword $\tildeEx{\cdot}$ and $h \in \AELC$ such that $\dis^L(\tildeEx{\cdot},h) \leq \delta_{dec}-\eps$, we can find $h$ in time $\calO(n) + \calT(n)$ with probability at least $\eps/\delta_{dec}$.
	\end{lemma}

	\begin{proof}
		First, we use the given pseudocodewords to build a collection of distributions $\calD = \inbraces{\calD_{\li}}_{\li \in L}$, each distribution over $[q_1]$. Recall that $\calC_0$ can be seen as a map from $[q_1]$ to $[q_0]^d$. The ${\li}^{th}$ distribution gives a weight of $\tildeEx{\indi{\zee_{N_L(\li)} = \calC_0(\alpha)}}$ to $\alpha \in [q_1]$. 

		Let $\overline{h}$ be the codeword in $\calC_1$ corresponding to the codeword $h$. That is, $\overline{h}$ is such that $\calC_0(\overline{h}_{\li}) = h_{N_L(\li)}$. With the collection of distributions $\calD$ defined, we relate $\dis^L(\tildeEx{\cdot},h)$ to agreement of $\overline{h}$ with $\calD$.
		\begin{align*}
			\dis^L(\tildeEx{\cdot},h) &= \Ex{\li}{\tildeEx{\indi{\zee_{N_L(\li)} \neq h_{N_L(\li)}}}} \\
			&= \Ex{\li}{\tildeEx{ \sum_{\alpha\in [q_1]} \indi{\zee_{N_L(\li)} = \calC_0(\alpha)} \indi{h_{N_L(\li)} \neq \calC(\alpha)} }} \\
			&=\Ex{\li}{\sum_{\alpha\in [q_1]} \indi{h_{N_L(\li)} \neq \calC_0(\alpha)} \tildeEx{ \indi{\zee_{N_L(\li)} = \calC_0(\alpha)}  }} \\
			&=\Ex{\li}{ \Ex{\alpha\sim \calD_{\li}}{ \indi{h_{N_L(\li)} \neq \calC_0(\alpha)}}} \\
			&=\Ex{\li}{ \Ex{\alpha\sim \calD_{\li}}{ \indi{\overline{h}_{\li} \neq \alpha}}}
		\end{align*}
		
		We call \cref{lem:decoding_from_distributions} for the code $\calC_1$ with the collection of distributions $\calD$ and find $\overline{h}$, and therefore $h$, with probability at least $\eps/\delta_{dec}$.
	\end{proof}
	
	To end this section, we note that the rounding from fractional vectors above can be derandomized through a standard method known as threshold rounding. Similar ideas are used to derandomize the classical Generalized Minimum Distance decoding for concatenated codes.
	
	\begin{lemma}\label{lem:derandomized_decoding_from_distributions}
	Let $\calC$ be an $[n,\delta,\rho]_q$ code, which is unique decodable from distance $\delta_{dec}\leq \delta/2$ in time $\calT(n)$. Given a collection of distributions $\calD = \inbraces{\calD_i}_{i\in [n]}$, each of them supported on $[q]$ described as a collection $q$ weights that sum to 1, there is a unique codeword $h$ that satisfies
	\[
		\Ex{i}{\Ex{j\sim \calD_i}{\indi{h_i \neq j}}} \le \delta_{dec}
	\]
	This codeword $h$ can be found in time $\calO(qn)\cdot\calT(n)$ with a deterministic algorithm.
	
	\begin{proof}
		The uniqueness of $h$ is as in proof of \cref{lem:decoding_from_distributions}.
		
		Let the weight on $j \in [q]$ according to $\calD_i$ be $w_{ij}$, so that $\sum_{j\in [q]} w_{ij} = 1$. We replace the randomized rounding of \cref{lem:decoding_from_distributions} by the following process:
		\begin{enumerate}[(i)]
			\item Define the cumulative sums $c_{ij} = \sum_{j' \leq j}{ w_{ij'}}$, so that $c_{iq} = 1$. Define $c_{i0} = 0$.
			\item For each $i\in [n]$, embed $\calD_i$ into the interval $[0,1]$ as $q+1$ points $(c_{i0} = 0, c_{i1}, c_{i2},\cdots , c_{iq} = 1)$.
			\item Choose $\theta \in [0,1]$ uniformly at random.
			\item Build $h'\in [q]^n$ coordinate-wise as follows. For $i\in [n]$, if 
			\[	
				\theta\in \left[c_{i(j-1)}, c_{ij}\right),
			\]
			then $h'_i = j$. This ensures that $\Pr{\theta}{h'_i = j} = w_{ij}$.
		\end{enumerate}
		We show that $h'$ has the same distance from $h$ in expectation as the collection of distributions $\inbraces{\calD_i}_{i\in [n]}$.
		\begin{align*}
			\Ex{\theta \in [0,1]}{\Delta(h',h)} &= \Ex{\theta \in [0,1]}{\Ex{i}{\indi{h'_i \neq h_i}}} \\
			&= \Ex{\theta \in [0,1]}{\Ex{i}{ \indi{\theta \not\in [c_{i(h_i-1)}, c_{ih_i}) }}} \\
			&= \Ex{i}{ \Ex{\theta \in [0,1]}{\indi{\theta \not\in [c_{i(h_i-1)}, c_{ih_i}) }}} \\
			&= \Ex{i}{1 - w_{ih_i}} \\
			&= \Ex{i}{\Ex{j\sim \calD_i}{\indi{h_i \neq j}}} \leq \delta_{dec}
		\end{align*}
		Therefore, rounding according to a random threshold $\theta$ produced an $h'$ that is at most $\delta_{dec}$ distance away from $h$. As the final step to derandomization, note that two thresholds $\theta_1,\theta_2$ produce the exact same $h'$ if there is no point from step (ii) above embedded between $\theta_1$ and $\theta_2$. Total number of points embedded is at most $q\cdot n$, and so it suffices to try at most $\calO(q\cdot n)$ many thresholds to produce all the different $h'$ possible - one of which must be $\delta_{dec}$-close to $h$.
	\end{proof}
\end{lemma}

\section{List Decoding Concatenated Codes}
\label{sec:concat}

In this section, we will adapt the techniques developed earlier to decode concatenated codes. Concatenation is a useful operation to obtain codes with smaller alphabet size starting from a code over large alphabet. Previous works on list decoding of concatenated codes \cite{GS00, GuruswamiR06, GuruswamiS02} seem to all rely on list recovery of outer code, with intricate weights to be passed along with inner codewords. We will only use list decodability of outer code.

First, we show that the covering lemma based argument can be used to decode the concatenated code up to the Johnson radius corresponding to product of decoding radius of outer code and the distance of inner code. That is, if the distance of inner code is $\delta_0$, distance of outer code is $\delta_1$, and the decoding radius of outer code is $\delta_{dec}$, we can decode the concatenated code up to radius $\calJ(\delta_{dec} \cdot \delta_0)$. 

Moreover, since concatenated codes do not involve expansion for their distance proof, we do not need to deal with SoS-based pseudocodewords or any low-covariance conditions. In fact, our pseudocodewords will just be local distributions over the inner code for each coordinate of the outer code. Since this set of pseudocodewords can be described as the feasible set corresponding to linear constraints over $\calO(n)$ variables, we can minimize the appropriate norm (which is a convex function) via Ellipsoid method to get a pseudocodeword with the covering property in time $n^{\calO(1)}$.

Note that this is weaker than decoding up to $\calJ(\delta_1\cdot\delta_0)$, which is the Johnson radius corresponding to the true distance of the concatenated code. In \cref{sec::concat_outer_ael}, we will see that by using the decoder of outer code in a white box way, we can get to the Johnson bound $\calJ(\delta_1\cdot\delta_0)$ when the outer code supports list decoding through our Covering Lemma based machinery, like in the case of near-MDS codes of \cref{thm:near_mds_main}.

\subsection{List decoding arbitrary concatenated codes}

Let the $[n,\delta_1,\rho_1]_{q_1}$ outer code be $\calC_1$  and the $[d,\delta_0,\rho_0]_{q_0}$ inner code be $\calC_0$, with $q_1 = q_0^{\rho_0 \cdot d}$. Let the concatenated code be $\calC^*$ with distance at least $\delta = \delta_0\cdot \delta_1$. A codeword $h^*=\calC_0(h)$ of the concatenated code can be seen as a tuple $h^* = (h^*_1,h^*_2,\cdots,h^*_n)$ where each $h^*_i\in \calC_0$, or $h^*_i = \calC_0(h_i)$ for some $h_i\in [q_1]$. Note that not all tuples of this form are codewords, and $(\calC_0(f_1),\calC_0(f_2),\cdots,\calC_0(f_n)) \in \calC^*$ iff $(f_1,f_2,\cdots,f_n)\in \calC_1$. 
	
\begin{definition}
	A pseudocodeword of the concatenated code is a psuedoexpectation operator $\tildeEx{\cdot}$ of degree $d$ over the variables $\zee = \inbraces{Z_{i,j,k}}_{i\in [n],j\in [d],k\in[q_0]}$ that respects the following $d$-local constraints:
	\begin{enumerate}[(i)]
		\item $Z_{i,j,k}^2 ~=~ Z_{i,j,k}$
		\item For every $i\in[n],j\in [d]$, $\sum_{k\in[q_0]} Z_{i,j,k} =1$
		\item	 $\forall i\in [n], \quad (\zee_{i,1},\zee_{i,2},\cdots,\zee_{i,d}) \in \calC_0$.
	\end{enumerate}
\end{definition}

We no longer enforce the constraint for non-negativity of squares of polynomials, and in fact, our pseudoexpectation operators are just a collection of $n$ distributions $\inbraces{\calD_i}_{i\in [n]}$ over $[q_1]$. The weight assigned to $f\in [q_1]$ in distribution $\calD_i$ is $\tildeEx{\indi{\zee_i = \calC_0(f)}}$.

Following is the natural generalization of distances to pseudocodewords.

\begin{definition}[Distance from a pseudocodeword]
The distance of a pseudocodeword $\tildeEx{\cdot}$ and a codeword $h^*$ is defined as
	\[
		\dis(\tildeEx{\cdot},h^*) = \Ex{i\in [n]}{\tildeEx{\dis(\zee_i,h^*_i)}}
	\]
\end{definition}

\begin{lemma}\label{lem:decode_from_pseudo}
	Assume that $\calC_1$ can be list-decoded from radius $\delta_{dec}$ in time $\calT(n)$ with list size $L$.
	
	For any $\eps>0$, there is a deterministic algorithm that given a pseudocodeword $\tildeEx{\cdot}$ outputs the list 
	\[
		\calL(\tildeEx{\cdot},\delta_0\cdot \delta_{dec}) := \inbraces{h^* \in \calC^* \suchthat \dis(h^*,\tildeEx{\cdot}) < \delta_0 \cdot \delta_{dec}}
	\]
	and runs in time $2^{\calO(d)} n^{\calO(1)}+\calO(2^d \cdot n)\cdot \calT(n)$.
\end{lemma}

\begin{proof}
	We use $\tildeEx{\cdot}$ to get local distributions $\calD_i$ over $[q_1]$ for each coordinate $i$. 
	\begin{align*}
		\dis(\tildeEx{\cdot},h^*) = \Ex{i\in [n]}{\tildeEx{\dis(\zee_i,h^*_i)}} &= \Ex{i\in [n]}{\Ex{f\sim \calD_i}{\dis \inparen{ \calC_0(f),h^*_i }}}\\
		&\geq \Ex{i\in [n]}{\Ex{f\sim \calD_i}{\indi{f\neq h_i}\dis \inparen{ \calC_0(f),h^*_i } }}\\
		&\geq \Ex{i\in [n]}{\Ex{f\sim \calD_i}{\indi{f\neq h_i} \delta_0}}\\
		&\geq \Ex{i\in [n]}{\Ex{f\sim \calD_i}{\indi{f\neq h_i}\dis \inparen{ \calC_0(f),h^*_i } }}\\
		&= \delta_0 \cdot \Ex{i\in [n]}{\Ex{f\sim \calD_i}{\indi{f\neq h_i}}}
	\end{align*}
	Any codeword $h^* =\calC_0(h) \in \calL$ must have the property that $\dis(\tildeEx{\cdot},h^*) <\delta_{dec} \cdot \delta_0$, so that
	\[
		\Ex{i}{\Ex{f\sim \calD_i}{\indi{f\neq h_i}}} < \delta_{dec}.
	\]
	Finally, we use \cref{lem:derandomized_decoding_from_distributions} to decode from the above agreement. The only modification is that since we might potentially deal with the \emph{list} decoding algorithm of outer code $\calC_1$, we take a union of all the lists generated by the different calls corresponding to different thresholds, and then prune it finally. For any $h^* = \calC_0(h)\in \calL$, there is some threshold for which the $[q_1]^n$ string generated in \cref{lem:derandomized_decoding_from_distributions} will be at distance $<\delta_{dec}$ from $h$. Therefore, $h$ will be contained in at least one of the lists discovered by the algorithm.
\end{proof}

\begin{lemma}\label{lem:concat_covering}
	For any $\eps>0$, there is an algorithm that given $g\in [q_0]^{nd}$ and $\delta>0$, finds a pseudocodeword $\tildeEx{\cdot}$ such that $\calL (g,\calJ(\delta)-\eps) \subseteq \calL(\tildeEx{\cdot},\delta-\eps_2)$ in time $2^{\calO(d)}\cdot n^{\calO(1)}$ for some $\eps_2>0$.
\end{lemma}

\begin{proof}[Proof Sketch]
	This is again an algorithmic implementation of the covering lemma, with the distributions over codewords to be relaxed to distributions over pseudocodewords as defined above. Minimizing the appropriate norm while optimizing over the convex set of pseudocodewords gives us this covering property. We omit the details since the argument is very similar to \cref{lem:cover_for_list_tanner}.
	
	Since the above optimization problem minimizes a convex quadratic function with linear constraints on at most $2^d\cdot n$ variables, we can get a running time of $2^{\calO(d)}\cdot n^{\calO(1)}$.
\end{proof}

\begin{theorem}
	Let $\calC_0$ be a binary inner code of blocklength $d$, distance $\delta_0$ and rate $\rho_0$. Also let $\calC_{1}$ be an outer code of blocklength $n$, distance $\delta_1$ and rate $\rho_{1}$ on an alphabet of size $|\calC_0|=2^{\rho_0d}$. Assume that $\calC_1$ can be list-decoded from radius $\delta_{dec}$ in time $\calT(n)$ with list size $L$.
	
	Then the code $\calC^*$ obtained by concatenating $\calC_1$ with $\calC_0$ can be list decoded up to a radius of $\calJ(\delta_{dec}\delta_0)$ in time $2^{\calO(d)}n^{\calO(1)} + \calO(2^d \cdot n) \cdot \calT(n)$.
\end{theorem}

\begin{proof}
	Given $g$ such that we wish to find $\calL = \calL(g,\calJ(\delta_{dec}\cdot \delta_0) -\eps)$, we first use \cref{lem:concat_covering} to find a pseudocodeword $\tildeEx{\cdot}$ that covers the list, and then use it via \cref{lem:decode_from_pseudo} to find $\calL$.
\end{proof}

\subsection{List Decoding with outer AEL codes}\label{sec::concat_outer_ael}

In this section, we show that when the near-MDS codes obtained using AEL distance amplification are used for concatenation with smaller alphabet codes, we can list decode the smaller alphabet code up to its Johnson bound. This will not be done by a black-box call to the list decoding/list recovery algorithm of the outer code, but we will crucially use the SoS-based list decoding strategy for the outer code.

To the best of our knowledge, list-decoding to the Johnson bound of the concatenated code has not been achieved for the Reed-Solomon outer code. This shows that while our near-MDS codes via AEL construction match Reed-Solomon codes in terms of list decoding radius, they have some extra desirable features. Of course, the runtime of our algorithms is quite poor compared to the near-linear time Reed-Solomon decoders.

The pseudocodewords for the concatenated code will be concatenations of the pseudocodewords of outer AEL code. Since the covering lemma works irrespective of the code we are working with, we can get a cover for the list of (final) codewords by efficiently optimizing over the pseudocodewords of outer code. Then a simple argument shows that the distance property for outer pseudocodewords translates to distance property for concatenated pseudocodewords.

Let us recall some notation for AEL Codes. Let $\AELC$ be an AEL code determined by an $(n,d,\lambda)$-expander graph $G=(L,R,E)$, an inner code $\calC_0$ of distance $\delta_0$, rate $r_0$, alphabet size $q_0$, and an outer code $\calC_1$ of distance $\delta_1$, rate $r_1$ and alphabet size $q_1 = |\calC_0|$. The code $\AELC$ is of alphabet size $q_0^d$, rate $r_0r_1$ and (designed) distance $\delta_0-\frac{\lambda}{\delta_1}$.

Suppose we concatenate $\AELC$ with an $[d_2,\delta_2,r_2]_{q_2}$ code $\calC_2$ such that $q_0^d = q_2^{r_2\cdot d_2}$, to obtain a new $\left[ nd_2, \inparen{\delta_0-\frac{\lambda}{\delta_1}} \delta_2, r_0r_1r_2\right]_{q_2}$ code $\calC_3$. Note that for the code $\calC_3$, distance is defined for $f_1,f_2\in [q_2]^{nd_2}$ as 
\begin{equation}\label{eqn:dist_for_concat}
	\dis^*(f_1,f_2) = \Ex{\ri\in R,j\in [d_2]}{\indi{f_1(\ri,j) \neq f_2(\ri,j)}}
\end{equation}
which can also be extended to distances between a pseudocodeword $\tildeEx{\cdot}$ and a codeword $h^*\in \calC_3$ as
\[
	\dis^*(\tildeEx{\cdot}, h^*) = \Ex{\ri\in R,j\in [d_2]}{\tildeEx{\indi{\calC_2(\zee_{N_R(\ri)})(j) \neq h^*(r,j) }}}
\]
If the codeword $h^*\in \calC_3$ is obtained by concatenating $h\in \AELC$ with $\calC_2$, we denote $h^* = \calC_2(h)$, and the above distance expression is the same as
\[
	\dis^*(\tildeEx{\cdot}, \calC_2(h)) = \Ex{\ri\in R,j\in [d_2]}{\tildeEx{\indi{\calC_2(\zee_{N_R(\ri)})(j) \neq \calC_2(h_{N_R(\ri)}) (j) }}}
\]
\begin{theorem}
	Let the code $\calC_1$ be decodable from radius $\delta_{dec}$, and $\lambda\leq \kappa \cdot \delta_{dec}$, so that the distance of code $\calC_3$ is at least $\delta_2(\delta_0-\kappa)$. For every $\eps > 0$, the code $\calC_3$ can be list decoded up to the radius $\calJ_{q_2} \inparen{ \inparen{\delta_0-\kappa} \delta_2 } - \eps$ in time $n^{\calO_{q,d,\delta_{dec}}(1/\eps^4)}$.
\end{theorem}

\begin{proof}[Proof Sketch]
We consider the same pseudocodewords as we did for AEL codes, and recall that we proved the following distance property for $\eta$-good pseudocodewords in \cref{lem:AEL_amplification}. For any $h\in \AELC$,

\[
	\Ex{\ri}{\tildeEx{\indi{\zee_{N_R(\ri)} \neq h_{N_R(\ri)}}}} \geq \delta_0 - \frac{\lambda+\eta}{\Ex{\li}{\tildeEx{\indi{\zee_{N_L(\li)} \neq h_{N_L(\li)}}}}}
\]
We extend the above distance property to the distance according to final $\calC_3$,
\begin{align*}
	\Ex{\ri\in R,j\in [d_2]}{\tildeEx{\indi{\calC_2(\zee_{N_R(\ri)})(j) \neq \calC_2(h_{N_R(\ri)})(j)}}}
&= \Ex{\ri}{\tildeEx{\dis\inparen{\calC_2(\zee_{N_R(\ri)}) , \calC_2(h_{N_R(\ri)})}}} \\
	&\geq \delta_2 \cdot \Ex{\ri}{\tildeEx{\indi{\zee_{N_R(\ri)} \neq h_{N_R(\ri)}}}} \\
	&\geq \delta_2 \inparen{ \delta_0 - \frac{\lambda+\eta}{\Ex{\li}{\tildeEx{\indi{\zee_{N_L(\li)} \neq h_{N_L(\li)}}}}} }
\end{align*}

That is,
\begin{equation}\label{eqn:concat_ael_distance}
	\dis^*\inparen{\tildeEx{\cdot}, \calC_2(h)} \geq \delta_2 \inparen{ \delta_0 - \frac{\lambda+\eta}{ \dis^L \inparen{\tildeEx{\cdot},h} }}
\end{equation}
The key distance property established, the rest of the argument is same as while decoding AEL codes. Given $g$ to be decoded, let's call the list of codewords at distance less than $\calJ(\delta_2(\delta_0-\kappa)) - \eps$ as $\calL$. We find a pseudocodeword $\tildeEx{\cdot}$ that is close to all the codewords in $\calL$. For any $h^*\in \calL$ such that $h^* = \calC_2(h)$,
\[
	\dis^*(\tildeEx{\cdot}, h^*) = \dis^*(\tildeEx{\cdot}, \calC_2(h)) < \delta_2(\delta_0-\kappa)-\eps_2
\]
where  $\eps_2 = 2\sqrt{1-\frac{q_2}{q_2-1} \delta_2(\delta_0-\kappa)}\cdot \eps = \Theta(\eps)$.

The covering lemma used here minimizes $\ell_2$-norm of the embedding corresponding to the concatenated code $\calC_3$ with the alphabet size $q_2$, while the relaxation for pseudocodeword was defined for $\AELC$ which does not depend on $\calC_2$ at all. However, this embedding according to $\calC_3$ is a linear function of $\tildeEx{\cdot}$, therefore it is still a convex function that is minimized.

By conditioning $\tildeEx{\cdot}$, we obtain another pseudocodeword $\dupPE{\cdot}$ which is $\eta$-good and retains its closeness to $\calC_2(h)$. By choosing $\lambda$ and $\eta$ small enough, \cref{eqn:concat_ael_distance} allows us to conclude that 
\[ 
	\dis^L \inparen{\tildeEx{\cdot},h} = \Ex{\li}{\tildeEx{\indi{\zee_{N_L(\li)} \neq h_{N_L(\li)}}}} < \delta_{dec}-\eps_3
\] 
for some $\eps_3>0$, which is sufficient to find $h$, and therefore $\calC_2(h)$, via \cref{lem:unique_decoding_ael}.

As before, this algorithm can be derandomized by going over all random choices which must discover the entire list $\calL(g,\calJ_{q_2}((\delta_0-\kappa)\delta_2)-\eps)$.
\end{proof}

Note that this argument does not place any restriction on $\delta_0$, and in particular, if we choose $\kappa$ to be much smaller than $\delta_0$, we can decode arbitrarily close to the Johnson radius corresponding to the product bound $\calJ(\delta_0 \cdot\delta_2)$, even for small values of $\delta_0$ like $1/3$. In contrast, the existing list decoding algorithms for concatenated codes via list recovery of outer Reed-Solomon code \cite{GuruswamiS02} only approach this Johnson bound when the outer distance is very close to 1.

In fact, we can also decode up to the Johnson bound of product of distances for any outer code that supports list decoding up to Johnson bound via our Covering Lemma/SoS-based machinery. This also includes Tanner code in particular.

\section{List Decoding Codes on Square Cayley Complex}

\label{sec:square_complex}

First, let us set some notation for the quadripartite left-right square Cayley complex on the group $G$ of size $n$ with generator sets $A,B$ of size $d$ each. Let $\calS$ denote the set of squares. The functions $X,Y,U,V$ are 4 bijections from $G\times A \times B$ to $\calS$, with the following property:
\[
	X(g,a,b) = Y(ga,a^{-1},b) = U(bga,a^{-1},b) = V(bg,a,b^{-1})
\]
The set $X(g,\cdot,b)\subseteq \calS$ is defined as $\{X(g,a,b) \suchthat a\in A\}$, and likewise for $X(g,\cdot,\cdot)$, $X(g,a,\cdot)$ and corresponding $Y,U,V$ sets. The sets $X(g,\cdot,\cdot),Y(g,\cdot,\cdot),U(g,\cdot,\cdot)$ and $V(g,\cdot,\cdot)$ should be seen as the analogs of sets $N_L(\li)$ and $N_R(\ri)$ from the Tanner codes.

The code on this square Cayley complex is then defined as 
\begin{align*}
	\calC^{SCC} = \{h \in [q]^{\calS} \suchthat \forall g\in G,\ & h|_{X(g,\cdot,\cdot)} \in C_A\otimes C_B, \\ 
	&h|_{Y(g,\cdot,\cdot)} \in C_A\otimes C_B, \\
	&h|_{U(g,\cdot,\cdot)} \in C_A\otimes C_B, \\
	&h|_{V(g,\cdot,\cdot)} \in C_A\otimes C_B\}
\end{align*}

where $C_A$ and $C_B$ are inner codes of blocklength $d$ each, and $\calC_A\otimes \calC_B$ is their tensor code.

For the code defined by left-right Cayley complex with inner codes $C_A, C_B$ with parameters $(d,\delta_A,r_A)$ and $(d,\delta_B,r_B)$ respectively, the distance of $\calC^{SCC}$ is lower bounded \cite{DELLM22} by \[ \delta = \delta_A \delta_B(\max(\delta_A,\delta_B) - \lambda)\]Note that the distance of tensor code $\calC_A\otimes \calC_B$ is at least $\delta_A \cdot \delta_B$. 

\begin{theorem}\label{thm:square-decoding}
For every $\eps>0$, there is an algorithm based on $\calO_{q,d}(1/\eps^4)$ levels of the SoS-hierarchy that runs in time $n^{\calO_{q,d}(1/\eps^4)}$ and can list decode $\calC^{SCC}$ up to $\calJ_q(\delta)- \eps$.
\end{theorem}

\begin{proof}[Proof Sketch]
	We outline the proof by once again focusing on a proof of distance for the appropriate notion of $\eta$-good pseudocodewords, and combining this proof with the covering lemma can be done as in the case of Tanner codes. Assume WLOG that $\delta_A\geq \delta_B$.
	
	Consider the SoS relaxation where variables correspond to squares: $\zee =\inbraces{Z_{s,j}}_{s\in \calS,j\in [q]}$. The SoS-degree of this relaxation is at least $d^2$ so that we can enforce all the inner code constraints by making the SoS relaxation respect such constraints explicitly.
	
	Note that for a fixed $b$, the set of squares $\{X(g,a,b) \suchthat g\in G, a\in A\}$ can be seen as the edges of an expander code, with the inner code as $C_A$. This motivates the following definition.
	
	A pseudocodeword is called $\eta$-good if for all $b \in B$,
	\begin{align}\label{eqn:eta_good_ltc}
		\Ex{g_1,g_2}{\tildeCov{\zee_{X(g_1,\cdot,b)}}{\zee_{Y(g_2,\cdot,b)}}} \leq \eta
	\end{align}
	
	As we proved before, this implies that for any $h\in \calC^{SCC}$, 
	\[
		\tau_b = \sqrt{\Ex{g}{\tildeEx{\indi{\zee_{X(g,\cdot,b)}\neq h_{X(g,\cdot,b)}}}}} \sqrt{\Ex{g}{\tildeEx{\indi{\zee_{Y(g,\cdot,b)}\neq h_{Y(g,\cdot,b)}}}}} 
	\] 
	satisfies
	\[
		\tau_b^2 -(\delta_A-\lambda)\tau_b+\eta \geq 0
	\]
	
	Let's prove upper and lower bounds on $\Ex{s}{\indi{Z_s\neq h_s}}$ in terms of $\tau_b$.
	
	\begin{align*}
		\Ex{s}{\indi{\zee_s \neq h_s}} &~=~ \Ex{g,a,b}{\tildeEx{\indi{\zee_{X(g,a,b)}\neq h_{X(g,a,b)}}}}  \\
		&~=~ \Ex{g}{\tildeEx{\dis(\zee_{X(g,\cdot,\cdot)}, h_{X(g,\cdot,\cdot)}}} \\
		&~\geq~ \Ex{g}{\tildeEx{0\cdot \indi{\zee_{X(g,\cdot,\cdot)} = h_{X(g,\cdot,\cdot)}}} + \delta_A \cdot \delta_B \cdot \indi{\zee_{X(g,\cdot,\cdot)} \neq h_{X(g,\cdot,\cdot)}}} \\
		&~=~ \delta_A\delta_B \Ex{g}{\tildeEx{\indi{\zee_{X(g,\cdot,\cdot )}\neq h_{X(g,\cdot,\cdot)}}}} \\
		&~\geq~ \delta_A\delta_B \Ex{g}{\tildeEx{\indi{\zee_{X(g,\cdot,b )}\neq h_{X(g,\cdot,b)}}}} 
	\end{align*}
	where the last inequality is true for any $b\in B$.
	
	Likewise, $\Ex{s}{\indi{\zee_s\neq h_s}} \geq \delta_A\delta_B \Ex{g}{\tildeEx{\indi{\zee_{Y(g,\cdot,b )}\neq h_{Y(g,\cdot,b)}}}}$. Therefore,
	
	\begin{align*}
		\Ex{s}{\indi{Z_s\neq h_s}} &\geq \delta_A\delta_B \sqrt{\Ex{g}{\tildeEx{\indi{\zee_{X(g,\cdot,b )}\neq h_{X(g,\cdot,b)}}}}} \sqrt{\Ex{g}{\tildeEx{\indi{\zee_{Y(g,\cdot,b )}\neq h_{Y(g,\cdot,b)}}}}} \\
		&= \delta_A\delta_B \tau_b
	\end{align*}
	
	For the upper bound,
	
	\begin{align*}
		\Ex{s}{\indi{Z_s \neq h_s}} &= \Ex{g,a,b}{\tildeEx{\indi{\zee_{X(g,a,b)}\neq h_{X(g,a,b)}}}} \\
		&= \Ex{b}{\Ex{g,a}{\tildeEx{\indi{\zee_{X(g,a,b)}\neq h_{X(g,a,b)}}}}} \\
		&\leq \Ex{b}{\tau_b^2 + \lambda \tau_b + \eta}
	\end{align*}
	
	Importantly, the lower bound works for any $b$, and so if any $\tau_b$ is large, we can conclude good distance. Otherwise, all $\tau_b$ are small, and then the upper bound works well.
	
	Suppose there is some $b$ for which $\tau_b \geq (\delta_A-\lambda) - \frac{2\eta}{\delta_A-\lambda}$. Then, the distance is at least $\delta_A\delta_B ( (\delta_A-\lambda) - \frac{2\eta}{\delta_A-\lambda})$.
	
	If not, then all $\tau_b$ are at most $\frac{2\eta}{\delta_A-\lambda}$. Then the distance,
	\begin{align*}
		\dis(\tildeEx{\cdot}, h) &\leq \Ex{b}{\tau_b^2 + \lambda\tau_b + \eta} \\
		&\leq \Ex{b}{\frac{4\eta^2}{(\delta_A-\lambda)^2} + \frac{\eta(\delta_A+\lambda)}{\delta_A-\lambda}} \\
		&= \frac{4\eta^2}{(\delta_A-\lambda)^2} + \frac{\eta(\delta_A+\lambda)}{\delta_A-\lambda}
	\end{align*}
	
	Putting everything together, if $\lambda\leq \delta_A/3$ and $\eta\leq \delta_A^2/9$, we get that 
	\[ \dis(\tildeEx{\cdot}, h) \geq \delta_A\delta_B(\delta_A-\lambda) - 3 \delta_B \eta \geq \delta -3\eta\]
	 or \[ \dis(\tildeEx{\cdot)}, h) \leq 3\eta\]
	
	Rest of the decoding argument is as before via covering lemma. We only need to show that the $\eta$-good property needed in \cref{eqn:eta_good_ltc} can be obtained by conditioning as before. In case we end up very close to a codeword, we can still unique decode via known unique decoding algorithms \cite{DELLM22} and \cref{lem:decoding_from_distributions}.
	
	The condition in \cref{eqn:eta_good_ltc} is that $d$ different average correlations are small. Recall that we ensured in the Tanner code case that covariance is small with probability $1-\gamma$ where $\gamma$ can be made arbitrarily small with SoS-degree. By a union bound, we can ensure that the covariance is small across all the $d$ bipartite graphs with probability at least $1-d\cdot \gamma$. Thus, by paying an additional factor of $d$ in the SoS-degree, we can ensure low covariance across all $b \in B$ as needed. 
	
	As before, this argument can be derandomized as well.
\end{proof}


\section*{Acknowledgements}
We thank Mary Wootters for help with references related to expander codes, and Mrinalkanti Ghosh for
helpful discussions. We are also grateful to FOCS 2023 reviewers for helpful comments on improving the presentation.

\bibliographystyle{alphaurl}
\bibliography{macros,madhur}

\end{document}